\documentclass[11pt]{article}
\usepackage[utf8]{inputenc}
\usepackage[margin=1in,a4paper]{geometry}
\usepackage{authblk}
\usepackage{hyperref}

\title{Near-Optimal Distributed Degree+1 Coloring\thanks{This paper incorporates results from the technical report \cite{HNT21} (by a subset of the authors of this paper) on $(\Delta+1)$-coloring in the {\LOCAL} model and is to be considered as the publication of that work. This excludes the additional results in \cite{HNT21} needed for a {\CONGEST} implementation, which will be published separately later.}
}

\newcommand{\CoAuthorMark}[1]{\footnotemark[#1]}
\author[1]{Magn\'us M. Halld\'orsson\thanks{Partially supported by Icelandic Research Fund grants 174484, 217965.}}
\author[2]{Fabian Kuhn}
\author[1]{Alexandre Nolin\protect\CoAuthorMark{2}}
\author[3]{Tigran Tonoyan\thanks{Partially supported by the European Union’s Horizon 2020 Research and Innovation Programme under grant agreement no. 755839.}}
\affil[1]{ICE-TCS \& Department of Computer Science, Reykjavik University, Iceland.}
\affil[2]{Department of Computer Science, University of Freiburg, Germany}
\affil[3]{Technion -- Israel Institute of Technology, Israel}

\affil[ ]{ \href{mailto:mmh@ru.is}{mmh@ru.is}; \href{mailto:kuhn@cs.uni-freiburg.de}{kuhn@cs.uni-freiburg.de}; \href{mailto:alexandren@ru.is}{alexandren@ru.is}; \href{mailto:ttonoyan@gmail.com}{ttonoyan@gmail.com}}

\usepackage{graphicx}
\usepackage{amsmath}
\usepackage{amsthm}
\usepackage{amssymb}
\usepackage{amsfonts}
\usepackage{algorithm}
\usepackage{algorithmic}
\usepackage{paralist}
\usepackage{color}
\usepackage{mathtools}
\usepackage{xspace}
\usepackage{hyperref}
\usepackage{subfig}

\usepackage{thmtools}
\usepackage{thm-restate} 

\usepackage{xcolor}

\makeatletter
\newcommand*{\bx}{bx}
\newcommand*{\IfBold}{
  \ifx\f@series\bx
    \expandafter\@firstoftwo
  \else
    \expandafter\@secondoftwo
  \fi
}
\makeatother

\newcommand{\alg}[2][]{{\IfBold{\MakeUppercase{#2}}{\textup{\textsc{#2}}}}{#1}\xspace}
\newcommand{\trycolor}[1][]{\alg[#1]{TryColor}}
\newcommand{\tryrandomcolor}{\alg{TryRandomColor}}
\newcommand{\computeacd}{\alg{ComputeACD}}

\newcommand{\disjointsample}[1][]{\alg[#1]{DisjointSample}}
\newcommand{\slackgeneration}[1][]{\alg[#1]{GenerateSlack}}
\newcommand{\putaside}[1][]{\alg[#1]{PutAside}}
\newcommand{\synchronizedcolortrial}{\alg{SynchColorTrial}}
\newcommand{\multitrial}[1][]{\alg[#1]{MultiTrial}}
\newcommand{\slackcolor}[1][]{\alg[#1]{SlackColor}}
\newcommand{\transversal}[1][]{\alg[#1]{Transversal}}
\newcommand{\lowdegreesample}[1][]{\alg[#1]{LowDegreeSample}}

\newcommand{\model}[1]{{\IfBold{\MakeUppercase{#1}}{\textup{\textsc{#1}}}}\xspace}
\newcommand{\LOCAL}{\model{Local}}
\newcommand{\CONGEST}{\model{Congest}}
\newcommand{\MPC}{\model{MPC}}
\newcommand{\eps}{\varepsilon}
\renewcommand{\epsilon}{\eps}

\DeclareMathOperator{\poly}{poly}

\definecolor{darkgreen}{rgb}{0,0.5,0}
\definecolor{darkblue}{rgb}{0,0,0.6}
\usepackage{hyperref}
\hypersetup{
    unicode=false,          
    colorlinks=true,        
    linkcolor=darkblue,          
    citecolor=darkgreen,        
    filecolor=magenta,      
    urlcolor=blue           
}
\usepackage[capitalize, nameinlink]{cleveref}
\Crefname{remark}{Remark}{Remarks}
\Crefname{observation}{Observation}{Observations}
\crefname{section}{Sec.}{Sec.}
\crefname{algorithm}{Alg.}{Alg.}

\newtheorem{theorem}{Theorem}
\newtheorem{lemma}{Lemma}
\newtheorem{claim}{Claim}
\newtheorem{corollary}{Corollary}
\newtheorem{definition}{Definition}
\newtheorem{observation}{Observation}
\newtheorem{proposition}{Proposition}

\DeclareMathOperator*{\Exp}{\mathbb{E}}

\DeclarePairedDelimiter{\abs}{\lvert}{\rvert}
\DeclarePairedDelimiter{\card}{\lvert}{\rvert}
\DeclarePairedDelimiter{\set}{\lbrace}{\rbrace}
\DeclarePairedDelimiter{\event}{\lbrack}{\rbrack}
\DeclarePairedDelimiter{\range}{\lbrack}{\rbrack}
\DeclarePairedDelimiter{\parens}{\lparen}{\rparen}
\DeclarePairedDelimiter{\ceil}{\lceil}{\rceil}

\newcommand{\knuthupuparrow}{\mathbin{\uparrow\uparrow}}

\newcommand{\bbI}{\mathbb{I}}

\newcommand{\BAD}{\ensuremath{\mathtt{BAD}}} 
\newcommand{\cshat}{c_{\mathrm{shatter}}}
\newcommand{\colSpace}{\mathcal{C}} 
\newcommand{\colspace}{\colSpace} 
\newcommand{\hcol}{\colSpace^{\mathrm{heavy}}}
\newcommand{\lcol}{\colSpace^{\mathrm{light}}}
\newcommand{\hcolS}[1]{\colSpace^{\mathrm{heavy},#1}}
\newcommand{\lcolS}[1]{\colSpace^{\mathrm{light},#1}}

\newcommand{\pal}{\Psi} 
\newcommand{\col}{\psi} 

\newcommand{\acset}{\mathcal{S}_{\mathrm{ac}}} 

\newcommand{\core}{I} 

\newcommand{\disc}{\bar{\eta}} 
\newcommand{\unev}{\eta} 
\newcommand{\barsigma}{\bar{\sigma}}
\newcommand{\spar}{\zeta} 

\newcommand{\rmsup}[2]{#1^{\mathrm{#2}}}

\newcommand{\discbal}{\rmsup{\disc}{bal}} 
\newcommand{\discunb}{\rmsup{\disc}{unb}} 
\newcommand{\sparbal}{\rmsup{\spar}{bal}} 
\newcommand{\sparunb}{\rmsup{\spar}{unb}} 
\newcommand{\Nbal}{\rmsup{N}{bal}} 
\newcommand{\Nunb}{\rmsup{N}{unb}} 
\newcommand{\dbal}{\rmsup{d}{bal}} 

\newcommand{\pgen}{p_{\mathrm{g}}} 
\newcommand{\pdisj}{p_{\mathrm{s}}} 

\newcommand{\cext}{c_{\mathrm{e}}} 
\newcommand{\cant}{c_{\mathrm{a}}} 

\newcommand{\smin}{s_{\min}}
\newcommand{\sminpow}{\rho}
\newcommand{\iratio}{\vartheta} 

\newcommand{\rmsub}[2]{#1_{\mathrm{#2}}}
\newcommand{\Vrm}[1]{\rmsub{V}{#1}}

\newcommand{\Vsp}{\Vrm{sparse}}  
\newcommand{\Vdense}{\Vrm{dense}} 
\newcommand{\Vun}{\Vrm{uneven}} 
\newcommand{\Vue}{\Vrm{ue}} 
\newcommand{\Vbal}{\Vrm{balanced}} 

\newcommand{\Vst}{\Vrm{start}} 
\newcommand{\Vstart}{\Vst} 
\newcommand{\Veasy}{\Vrm{easy}} 
\newcommand{\Vheavy}{\Vrm{heavy}} 
\newcommand{\Vtough}{\Vrm{tough}} 
\newcommand{\Vdisc}{\Vrm{disc}}
\newcommand{\Vhe}{\rmsup{V}{he}} 

\newcommand{\UP}{A_v}
\newcommand{\UE}{B_v}
\newcommand{\bUP}{\tilde{\UP}}

\newcommand{\epsrm}[1]{\rmsub{\eps}{#1}}
\newcommand{\ehc}{\epsrm{hc}}  
\newcommand{\eub}{\epsrm{ub}}  
\newcommand{\eacd}{\epsrm{ac}}  
\newcommand{\espa}{\epsrm{sp}}  
\newcommand{\ehat}{\hat{\eps}}  

\newcommand{\DeloLC}{\ensuremath{\mathrm{\Delta1LC}}\xspace}
\newcommand{\DeloC}{\ensuremath{\mathrm{\Delta1C}}\xspace}
\newcommand{\degoLC}{\ensuremath{\mathrm{D1LC}}\xspace}
\newcommand{\degoC}{\ensuremath{\mathrm{D1LC}}\xspace}

\begin{document}

\pagenumbering{roman}
\begin{titlepage}

\date{}

\maketitle

\begin{abstract}
  We present a new approach to randomized distributed graph coloring that is simpler and more efficient than previous ones. In particular, it allows us to tackle the $(\deg+1)$-list-coloring ($\degoLC$) problem, where each node $v$ of degree $d_v$ is assigned a palette of $d_v+1$ colors, and the objective is to find a proper coloring using these palettes. While for $(\Delta+1)$-coloring (where $\Delta$ is the maximum degree), there is a fast randomized distributed $O(\log^3\log n)$-round algorithm (Chang, Li, and Pettie \cite{CLP20}), no $o(\log n)$-round algorithms are known for the $\degoLC$ problem.
  
We give a randomized distributed algorithm for $\degoLC$
that is optimal under plausible assumptions about the deterministic complexity of the problem.  Using the recent deterministic algorithm of Ghaffari and Kuhn~\cite{GK21}, our algorithm runs in $O(\log^3 \log n)$ time, matching the best bound known for $(\Delta+1)$-coloring. In addition, it colors all nodes of degree $\Omega(\log^7 n)$ in $O(\log^* n)$ rounds.
  
  A key contribution is a subroutine to generate slack for $\degoLC$. When placed into the framework of Assadi, Chen, and Khanna \cite{ACK19} and Alon and Assadi \cite{AA20}, this almost immediately leads to a palette sparsification theorem for $\degoLC$, generalizing the results of \cite{ACK19,AA20}. That gives fast algorithms for $\degoLC$ in three different models: an $O(1)$-round algorithm in the \MPC model with $\tilde{O}(n)$ memory per machine; a single-pass semi-streaming algorithm in dynamic streams; and an $\tilde{O}(n\sqrt{n})$-time algorithm in the standard query model.
\end{abstract}
\thispagestyle{empty}
\end{titlepage}

{   \newpage
    \smallskip
    \hypersetup{linkcolor=darkblue}
    \tableofcontents
    \setcounter{page}{2}
    \thispagestyle{empty}
}
\newpage
\pagenumbering{arabic}

\section{Introduction and Related Work}
\label{sec:intro}

The distributed vertex coloring problem is one of the defining and probably the most intensively studied problem of the area of distributed graph algorithms. In the standard version of the problem we are given a graph $G=(V,E)$ that at the same time defines a communication network and also the graph to be colored. In a distributed coloring algorithm, the nodes $V$ of $G$ communicate with each other in synchronous rounds by exchanging messages over the edges $E$ of $G$. Initially, the nodes of $G$ do not know anything about $G$ (except possibly for some global parameters such as, for example, the number of nodes $n$, the maximum degree $\Delta$, or approximations thereof), and at the end, each node $v\in V$ must output a color such that adjacent nodes are colored with different colors and such that the overall number of colors is from a given restricted domain. If adjacent nodes can exchange arbitrarily large messages in every communication round, this distributed model is known as the \LOCAL model and if messages are restricted to $O(\log n)$ bits per edge in each round, the model is known as the \CONGEST model~\cite{peleg00}. The time or round complexity of a distributed algorithm in the \LOCAL or \CONGEST model is the total number of rounds until all nodes terminate.

\paragraph{Early work on distributed coloring.} The distributed coloring problem was first studied in a seminal paper by Linial~\cite{linial92}, which essentially also started the whole area of local distributed graph algorithms. Linial showed in particular that any deterministic distributed algorithm for computing an $O(1)$-coloring of a ring network requires $\Omega(\log^* n)$ rounds. He also showed that in $O(\log^* n)$ rounds, it is possible to (deterministically) color arbitrary graphs of maximum degree $\Delta$ with $O(\Delta^2)$ colors. The $\Omega(\log^* n)$ lower bound was later extended to randomized algorithms by Naor~\cite{Naor91}. With a simple sequential greedy algorithm, one can color the vertices of a graph with at most $\Delta+1$ colors, and most of the work on distributed coloring was therefore also on solving the $(\Delta+1)$-coloring problem. Already when Linial's paper came out in 1987, it was clear that the randomized parallel maximal independent set algorithms developed shortly before by Luby~\cite{luby86} and Alon, Babai, and Itai~\cite{alon86} can be used to obtain a randomized distributed $O(\log n)$-round algorithm to compute a $(\Delta+1)$-coloring. In fact, even the na\"{\i}ve parallel  coloring algorithm, where each node repeatedly chooses a uniformly random color among the still available colors and keeps the color if no neighbor concurrently tries the same color, leads to an $O(\log n)$-round distributed $(\Delta+1)$-coloring algorithm~\cite{johansson99}.

\paragraph{A brief history on distributed \boldmath$(\Delta+1)$-coloring algorithms.} Given that there are very simple $O(\log n)$-time randomized distributed $(\Delta+1)$-coloring algorithms, until relatively recently, most of the work was on deterministic distributed coloring algorithms. Given a coloring with more than $\Delta+1$ colors, it is straightforward to reduce the number of colors by one in a single round. The $O(\log^* n)$-time $O(\Delta^2)$-coloring algorithm of Linial~\cite{linial92} therefore directly leads to an $O(\Delta^2 + \log^* n)$-time distributed algorithm for $(\Delta+1)$-coloring and thus in bounded-degree graphs, a $(\Delta+1)$-coloring can be computed in optimal $O(\log^* n)$ rounds. Over the years, the dependency on $\Delta$ has been improved in a long sequence of papers to the current best algorithm, which has a round complexity of $O(\sqrt{\Delta\log\Delta} + \log^* n)$~\cite{goldberg88,SzegedyV93,KuhnW06,Kuhn2009WeakColoring,BarenboimEK14,Barenboim16,FHK,BEG18,MT20}. For the $(2\Delta-1)$-edge coloring problem (i.e., for the same problem on line graphs), the time complexity has even been improved recently to $(\log\Delta)^{O(\log\log\Delta)}+O(\log^* n)$~\cite{BE11_neighborhoodInd,Kuhn20,BKO20}. 

As a function of the number of nodes $n$, the fastest known deterministic algorithms have long been based on computing a so-called network decomposition (a decomposition of the graph into clusters of small diameter together with a coloring of the clusters with a small number of colors). Until a recent breakthrough by Rozho\v{n} and Ghaffari~\cite{RG20}, the best deterministic algorithm for computing such a network decomposition and the best resulting $(\Delta+1)$-coloring algorithm had a round complexity of $2^{O(\sqrt{\log n})}$~\cite{awerbuch89,panconesi1992improved}. Rozho\v{n} and Ghaffari~\cite{RG20} improved this to $\poly\log n$ rounds. When focusing on the dependency on $n$, there is also work on computing vertex and edge colorings directly, without going through network decomposition~\cite{barenboimE10,FGK17,FOCS18-derand,HarrisEdge19,Kuhn20,GK21}. This has culminated in the recent work of Ghaffari and Kuhn~\cite{GK21}, who showed that a $(\Delta+1)$-coloring can be computed in $O(\log^2\Delta\cdot\log n)$ rounds deterministically. The algorithm of \cite{GK21} also works directly in the \CONGEST model. 

In light of the simple $O(\log n)$-time randomized distributed $(\Delta+1)$-coloring algorithm from the late 1980s, work on faster randomized distributed coloring algorithms only started a bit more than 10 years ago. In \cite{KSOS06}, it was shown that computing an $O(\Delta)$-coloring can be done in $O(\sqrt{\log n})$ rounds and in \cite{SW10}, this was even improved to $O(\log^* n)$ as long as $\Delta\geq \log^{1+\eps}n$. As one of the results of the current paper, we show that for $\Delta\geq \log^{2+\eps}n$, also a $(\Delta+1)$-coloring can be computed in only $O(\log^* n)$ rounds. The first improvements on the complexity of the $(\Delta+1)$-coloring problem were obtained in \cite{SW10,BEPSv3} and the first sub-logarithmic-time algorithms for $(2\Delta-1)$-edge coloring  and for $(\Delta+1)$-vertex coloring were subsequently developed in \cite{EPS15,HSS18}. This development led to the algorithm of Chang, Li, and Pettie~\cite{CLP20}, which in only $O(\log^* n)$ rounds manages to compute a partial $(\Delta+1)$-vertex coloring such that all remaining uncolored components are of polylogarithmic size. In combination with the deterministic algorithm of \cite{GK21}, this leads to a randomized $(\Delta+1)$-coloring algorithm with a round complexity of $O(\log^3\log n)$. An adaptation of the algorithm of \cite{CLP20} to the \CONGEST model appeared in \cite{HKMT21}. We will provide a more detailed discussion of the papers \cite{SW10,BEPSv3,EPS15,HSS18,CLP20,HKMT21} that are most relevant for the present work in \Cref{sec:techoverview}.

\paragraph{From $(\Delta+1)$-coloring to \boldmath$(\mathrm{degree}+1)$-list-coloring.} In a $c$-list-coloring problem, each node $v$ is given as input a list or palette consisting of $c$ colors from some color space $\colSpace$, and the objective is to compute a proper coloring of the graph, where each node $v$ is colored with a color from its list. In the $(\mathrm{degree}+1)$-list coloring (\degoLC) problem, the list of each node is of size (at least) $d_v+1$, where $d_v$ is the (initial) degree of $v$. The \degoLC problem is a natural generalization of the $(\Delta+1)$-coloring problem that can still be solved by the na\"{\i}ve sequential greedy algorithm. Further, after computing a partial solution to a given $(\Delta+1)$-coloring problem, the remaining coloring problem on the uncolored nodes in general is a \degoLC problem, where the palette of each node consists of the colors not used by any of the neighbors. In some sense, the \degoLC problem is the more fundamental and also the more natural problem than the $(\Delta+1)$-coloring problem. The \degoLC problem is self-reducible: After computing a partial solution to a given \degoLC problem, the remaining problem is still a \degoLC problem. Also the \degoLC problem naturally appears as a subproblem when solving more constrained coloring problems. \degoLC is, for example, used as a subroutine in the distributed coloring algorithms of \cite{BamasEsperet19} for computing optimal colorings in graphs with chromatic number close to $\Delta$ and in the distributed $\Delta$-coloring algorithms of \cite{PanconesiS95,GHKM18}.

\paragraph{Distributed \degoLC algorithms.} First note that all the fastest deterministic and randomized $(\Delta+1)$-coloring algorithms discussed above also work for the more general $(\Delta+1)$-list coloring problem. In fact, many of those algorithms critically rely on the fact that they solve some version of the list coloring problem, e.g., \cite{Barenboim16,FHK,MT20,HSS18,CLP20,Kuhn20,BKO20}. Further, by using techniques developed in \cite{FHK,Kuhn20}, one can deterministically reduce $\degoLC$ to $(\Delta+1)$-list coloring with only an $O(\log\Delta)$ multiplicative and $O(\log^* n)$ additive overhead.\footnote{If the round complexity is polynomial in $\Delta$, the multiplicative dependency even reduces to $O(1)$.} For deterministic algorithms, at least currently, there is therefore no significant gap between the complexities of $(\Delta+1)$-(list)-coloring and \degoLC. This is however very different for randomized algorithms. While the best known $(\Delta+1)$-list coloring algorithm requires only $O(\log^3\log n)$ rounds~\cite{CLP20}, the best known randomized algorithm that works for the \degoLC problem is from \cite{BEPSv3} and it has a round complexity of $O(\log\Delta + \log^3\log n)$. For general graphs, this can be as large as $O(\log n)$ and it is therefore not faster than the simple randomized $O(\log n)$-time distributed coloring algorithms~\cite{alon86,luby86,linial92,johansson99} from the 1980s and 1990s (those algorithms also work  directly for the \degoLC problem).

\subsection{Our Contributions}

The main technical contribution of our paper is an $O(\log^* n)$-time randomized distributed algorithm that for a given \degoLC problem, colors almost all nodes of an exponentially large degree range. More concretely, we prove the following technical theorem.

\begin{restatable}{theorem}{maintechtheorem}\label{thm:main_tech}
    Let $G=(V,E)$ be an $n$-node graph with maximum degree at most $\Delta$ and let $V_H$ be the nodes of $G$ of degree at least  $\log^7\Delta$. Then, for every positive constant $c>0$, there is an $O(\log^* n)$-round randomized distributed algorithm that for a given \degoLC instance on $G$ computes a partial proper coloring of the nodes in $V_H$ such that for every node $v\in V_H$, the probability that $v$ is not colored at the end of the algorithm is at most $1/\Delta^c$, even if the random bits of nodes outside the $2$-hop neighborhood of $v$ are chosen adversarially.
\end{restatable}

Our main contribution follows from \cref{thm:main_tech} using standard techniques. By applying methods originally used by Beck in the context of algorithmic versions of the Lov\'{a}sz Local Lemma~\cite{beck1991algorithmic} and first adapted to the distributed context in \cite{BEPSv3}, the probabilistic guarantees of the above theorem imply that after running the $O(\log^* n)$-round randomized distributed algorithm, w.h.p., the uncolored nodes form components of size $\poly(\Delta\log n)$. This phenomenon is nowadays known as graph shattering. One can now go through $O(\log^* n)$ degree classes. If we set $\Delta=n$, \cref{thm:main_tech} implies that all nodes $v$ of degree $d_v\in[\log^7 n,n]$ can be colored in $O(\log^* n)$ rounds w.h.p. For lower degree classes, the $O(\log^* n)$-round algorithm colors all nodes, except for components of $\poly\log n$ size. To color those components, one can then apply the best deterministic \degoLC algorithm of \cite{GK21}, which has a round complexity of $O(\log^2\Delta \log N)$ on $N$-node graphs of maximum degree $\Delta$ and thus a round complexity of $O(\log^3\log n)$ on graphs of size $\poly\log n$. Overall, we obtain the following main theorem.

\begin{restatable}{theorem}{maintheorem}\label{thm:main}
    There is a randomized distributed algorithm to solve the \degoLC problem on $n$-node graphs in $O(\log^3\log n)$ rounds, w.h.p.
\end{restatable}

The fact that our randomized algorithm directly colors all nodes of degree at least $\log^7 n$ has another interesting consequence. The following is a direct corollary of \cref{thm:main_tech}.

\begin{restatable}{corollary}{logstarcorollary}\label{cor:largeDelta}
When all nodes have degree at least $\log^7 n$, the \degoLC problem can be solved w.h.p.\ in $O(\log^* n)$ rounds in the \LOCAL model. 
\end{restatable}

Note that \cref{cor:largeDelta} is a significant improvement over prior work. Prior to this paper, also for large $\Delta$, the best known $(\Delta+1)$-coloring had a round complexity of $O(\log^3\log n)$ (even for the standard non-list version of the problem). Also note that the statement of \cref{cor:largeDelta} can be obtained by a somewhat simpler algorithm and by a much simpler analysis than the full statement of \cref{thm:main_tech}.

We show in the appendix that the lower bound on the degrees can be reduced in the case of the $\DeloLC$ problem.

\begin{restatable}{corollary}{logstarcorollarydelta}\label{cor:largeDeltaDelo}
When $\Delta \ge \log^{2+\delta} n$ for $\delta > 0$, the $\Delta+1$-(list)-coloring problem can be solved w.h.p.\ in $O(\log^* n)$ rounds in the \LOCAL model. 
\end{restatable}

\paragraph{Palette sparsification}

One key technical lemma is a method to generate slack.
One corollary of that result is the following result.

\begin{theorem}[Informal]
\label{thm:palettesparsify}
For any graph $G(V,E)$, sampling $O(\log^2 n)$ colors for each vertex with degree $d_v$ from a set of $d_v+1$ arbitrary colors, allows for a proper coloring of $G$ from the sampled colors, w.h.p.
\end{theorem}
This was previously shown for $\Delta+1$-coloring \cite{ACK19}, $\deg+1$-coloring \cite{AA20}, and $(1+\epsilon)\deg$-list-coloring \cite{AA20}, but in all cases requiring only $O(\log n)$-sized samples (which are necessary).
Our result follows almost immediately from the frameworks of \cite{ACK19,AA20} when given the slack generation result for sparse nodes (\cref{L:slackgen-sparse}).

This has the following implication for the \degoLC problem in several other models.
\begin{corollary}
For finding a $\degoLC$ in a general graph, w.h.p., there exists
\begin{compactenum}
\item a single-pass dynamic streaming algorithm using $\tilde{O}(n)$ space;
\item a non-adaptive $\tilde{O}(n^{3/2})$-time algorithm; and
\item an \MPC algorithm in $O(1)$-rounds on machines with memory $\tilde{O}(n)$.
\end{compactenum}
\end{corollary}

We discuss these implications in \cref{sec:palette}.

\section{Technical Overview}
\label{sec:techoverview}

In the following, we first discuss the most important technical insights that lead to the current fast randomized distributed $(\Delta+1)$-coloring algorithms. We next highlight why the existing techniques are not sufficient to also solve the $(\deg+1)$-coloring ($\degoLC$) problem similarly efficiently, and where the main challenges are. We then give a high-level overview on how we overcome those challenges and at the same time also simplify the existing randomized distributed $(\Delta+1)$-coloring algorithms.

\paragraph{Graph shattering.} The graph shattering technique was pioneered by Beck~\cite{beck1991algorithmic} in the context of constructive solutions for the Lov\'{a}sz Local Lemma, and it was brought to the distributed setting by Barenboim, Elkin, Pettie, and Schneider~\cite{BEPSv3}. The high-level idea is the following: One first runs a fast randomized algorithm that computes a partial solution for a given graph problem such that the unsolved parts only form small components (i.e., the randomized algorithm shatters the graph into small unsolved components). The remaining small components are then typically solved by a deterministic algorithm. More formally, let $G=(V,E)$ be an $n$-node graph of maximum degree $\Delta$ and assume that a randomized distributed algorithm computes an output for a (random) subset $S\subseteq V$ of the nodes. If for every node $v\in V$, independently of the private randomness of nodes outside some constant neighborhood of $v$, $\Pr(v\not\in S)\leq 1/\Delta^k$ for a sufficiently large constant $k$, then, w.h.p., the induced subgraph $G[V\setminus S]$ of the nodes with no output consists of connected components of size at most $\poly(\Delta \cdot \log n)$. A formal statement of this appears, e.g.,\  in \cite{FOCS18-derand,CLP20}. With some additional tricks (or in the case of graph coloring, often even directly), the size of the remaining components can be reduced to $\poly\log n$, so that the randomized complexity of a problem becomes the time to shatter the graph plus the time to solve the remaining problem deterministically on graphs of size $\poly\log n$. Interestingly, it was shown by Chang, Kopelowitz, and Pettie~\cite{CKP19} that the randomized distributed complexity of all locally checkable labeling problems (to which all the typical  coloring problems belong) on graphs of size $n$ is at least the deterministic complexity of the same problem on instances of size $\sqrt{\log n}$. The graph shattering method is therefore essentially necessary for solving such problems.

\paragraph{The role of slack.} At the core of all sublogarithmic-time randomized distributed (list) coloring algorithms is the notion of slack. A node $v$ of degree $d(v)$ is said to have slack $s(v)$ if it has an available color palette (or list) $\Psi(v)$ of size $|\Psi(v)|\geq d(v)+s(v)$. If we are given a coloring problem in which all nodes $v$ have slack $s(v)=\Omega(d(v))$, one can use an idea of Schneider and Wattenhofer~\cite{SW10} to color (most of) the graph in only $O(\log^* n)$ rounds as follows. Assume that for each node $v$, $s(v)/d(v)\geq \alpha$, for some $\alpha>0$. Each node $v$ chooses $\Theta(\alpha)$ random color from its list $\Psi(v)$ of colors, and $v$ gets permanently colored with one of those colors if no neighbor tries the same color. For each node $v$, each of the $\Theta(\alpha)$ colors has a constant probability of being successful and therefore each node gets permanently colored with probability $1-e^{-\Theta(\alpha)}$. In the remaining coloring problem on the uncolored nodes, the degree of most nodes drops by a factor $e^{\Theta(\alpha)}$, while the slack $s(v)$ of a node cannot decrease. The slack to degree ratio of most nodes therefore increases from $\alpha$ to $e^{\Theta(\alpha)}$. If we start with slack $s(v)=\Omega(d(v))$ and thus $\alpha=\Omega(1)$, after only $O(\log^* n)$ rounds, most nodes are permanently colored with a color from their list. In the $(\Delta+1)$-coloring problem, high-degree nodes however do not initially start with sufficient slack. In the $\degoLC$ problem, all nodes start with a color palette of size $d(v)+1$ and thus with slack $s(v)=1$. If we want to apply the above fast coloring algorithm in those cases, we first have to create slack for nodes.

\paragraph{Basic slack generation for $(\Delta+1)$-coloring.} In principle, there are three ways of generating slack for a node and we use all three ways in our algorithm. The slack $s(v)$ of a node increases if some neighbor $u$ permanently chooses a color $x\in\Psi(u)\setminus\Psi(v)$ that is not in $v$'s palette (we will refer to this as \emph{chromatic slack}) and it also increases if there are two (non-adjacent) neighbors $u$ and $w$ of $v$ that both permanently choose the same color. In addition, the slack $s(v)$ of a node $v$ can be temporarily increased if the nodes are colored in different phases of an algorithm and some neighbors of $v$ are colored in a later phase than $v$. In the $(\Delta+1)$-list-coloring problem, slack can be generated for many nodes by applying the following simple one-round distributed algorithm. Each node $v$ tries a uniformly random color of its palette $\Psi(v)$ and $v$ is permanently colored with this color if no neighbor of $v$ tries the same color. Because all nodes choose from $\Delta+1$ different colors, it is not hard to see that every node has a constant probability of keeping the color it tries. In expectation (and with sufficiently high probability), node $v$ therefore gets slack $s(v)=p\cdot d(v)$ if either the average probability for neighbors to pick a color outside $\Psi(v)$ is at least $p$ or if there are $\Theta(p\cdot d_v)$ non-connected pairs that try the same color (note that each color is only tried a constant number of times by nodes in $N(v)$ in expectation).

\paragraph{Almost-clique decomposition.}
All known sublogarithmic-time distributed $(\Delta+1)$-(list)-coloring algorithms are based on the following high-level idea. As a first step, the nodes are partitioned into a set $\Vsp$ of nodes that are \emph{locally sparse} and into so-called \emph{almost-cliques}. A node $v$ is said to have sparsity $\zeta$ if the subgraph induced by its neighborhood $N(v)$ contains at most $\binom{\Delta}{2}-\zeta\Delta$ edges. In an almost-clique decomposition, for some parameter $\eps>0$, the nodes in $\Vsp$ have sparsity $\Omega(\eps^2\Delta)$ and each almost-clique $C$ is a set of nodes for which $|C|\leq (1+\eps)\Delta$ and each node in $C$ has at least $(1-\eps)\Delta$ neighbors in $C$. A similar decomposition was first used by Reed~\cite{Reed98} and it was first used in the context of distributed coloring by Harris, Schneider, and Su~\cite{HSS18}. Since then, most fast randomized coloring algorithms in the distributed setting or related computational models are based on almost-clique decompositions~\cite{CLP20,HKMN20,HKMT21,ParterSu,CDP20,AA20,ACK19,CFGUZ19}. 

For locally sparse nodes, the required condition for slack generation described in the paragraph above is satisfied. One can therefore first let every node try a random color and let nodes keep their colors if no neighbor chooses the same color. The uncolored locally sparse nodes in this way get some slack and we can then delay coloring them to the end of the algorithm. The almost-cliques can in principle be handled efficiently because any two nodes within a single almost-clique are within distance $2$ in the graph. At least in the \LOCAL model, computations within a single almost-clique can therefore be done in a centralized fashion. Note however that implementing this high-level idea is not trivial. If $\eps$ is chosen large (e.g., as a small constant), the locally sparse nodes get a lot of slack and can be colored very fast, but this also creates a lot of dependencies between the different almost-cliques. If $\eps$ is small, the dependencies between almost-cliques become easier to handle, while now the locally sparse nodes also obtain less slack. In \cite{HSS18}, the authors set $\eps$ to balance the time for coloring the almost-cliques and for afterwards coloring the locally sparse nodes. The algorithm was then improved by Chang, Li, and Pettie in a technical tour de force~\cite{CLP20}. The authors of \cite{CLP20} define (and construct) a hierarchy of almost-cliques with different $\eps$ and they show that this hierarchy can be used to shatter the graph in only $O(\log^* n)$ rounds, which in combination with the deterministic algorithm of \cite{GK21} leads to the current fastest $O(\log^3\log n)$-round distributed $(\Delta+1)$-coloring algorithm. The approach of \cite{CLP20} was simplified and adapted to the \CONGEST model in \cite{HKMT21}. In \cite{HKMN20}, it was in particular shown (in the context of the more constrained \emph{distance-2 coloring problem} in \CONGEST) that one can compute a single almost-clique decomposition for a constant $\eps>0$ and that after running one round in which every node tries to get colored with a random color of its list, each node $v$ in an almost-clique $C$ obtains slack proportional to the number of neighbors $v$ has outside $C$ with large probability. This was used in \cite{HKMT21} to color the almost-cliques in $O(\log\log\Delta)$ rounds of the \CONGEST model.

\paragraph{Extending the setup to \boldmath$(\deg+1)$-list coloring (\degoLC).} When extending existing randomized $(\Delta+1)$-coloring algorithms to the more restrictive \degoLC problem, one faces a number of challenges. First, the notion of local sparsity and the almost-clique decomposition have mostly been defined for the $(\Delta+1)$-coloring problem~\cite{Reed98,EPS15,CLP20,ACK19,HKMT21,HNT21}: a node is locally sparse if the number of edges among neighbors is small compared to a complete neighborhood of size $\Delta$ and almost-cliques have to be of size close to $\Delta$. Luckily, Alon and Assadi~\cite{AA20} gave a generalization of the almost-clique decomposition that can be used for the $\degoLC$ problem. The decomposition is mostly defined in a natural way. The definitions of local sparsity and almost-cliques are now w.r.t.\ to the actual node degrees instead of w.r.t.\ $\Delta$ and the authors in addition define a node $v$ to be uneven if a constant fraction of the neighbors of $v$ have a sufficiently higher degree. They then show that the nodes of the graph can be partitioned into a set of locally sparse nodes, a set of uneven nodes, and several almost-cliques. As the more standard almost-clique decompositions, this decomposition can be computed in constant time in the \LOCAL model.

Based on the generalized almost-clique decomposition for the $\degoLC$ problem, we would like to proceed in a similar way as for the $(\Delta+1)$-coloring problem. As a first step, we would like to create slack for all nodes that are not in almost-cliques, i.e., for all nodes that are locally sparse and for all nodes that are uneven. The major obstacle that we have to overcome to achieve this is the problem of generating slack. This was already pointed out by Chang, Li, and Pettie~\cite{CLP20} as a major obstacle to the generalization of their result to the \degoLC problem. In fact, \cite{CLP20} suggests to first look at the simpler \emph{$(\deg+1)$-coloring} problem, where a node $v$ of degree $d(v)$ is to be assigned a color from $\set{1,\dots,d(v)+1)}$.

\paragraph{Slack generation for \boldmath$\deg+1$-list-coloring.}
The \degoLC problem brings a number of challenges for slack generation that are not present in the $(\Delta+1)$-list coloring problem. In the $(\Delta+1)$-(list)-coloring problem, nodes of degree $(1-\eps)\Delta$ have slack more than $\eps\Delta$ from the start because their palettes are of size $\Delta+1$. It is further well-established that a node $v$ of high degree and sufficiently large local sparsity obtains slack by a single round of trying a random color. Intuitively, this follows because the palettes of non-adjacent neighbors of $v$ either have a large overlap, leading to slack via sparsity, or they contain many colors that are not in $v$'s palette, leading to chromatic slack. In the $\degoLC$ problem, neither low-degree nodes nor locally sparse high-degree nodes get automatic slack. To illustrate the problems that can arise, we examine a few motivating examples.

\begin{figure}[th]
    \centering
    \subfloat[\label{subfig:hard-sparse}A sparse node adjacent to two cliques with non-overlapping palettes.]{\includegraphics[width=.45\textwidth,page=1]{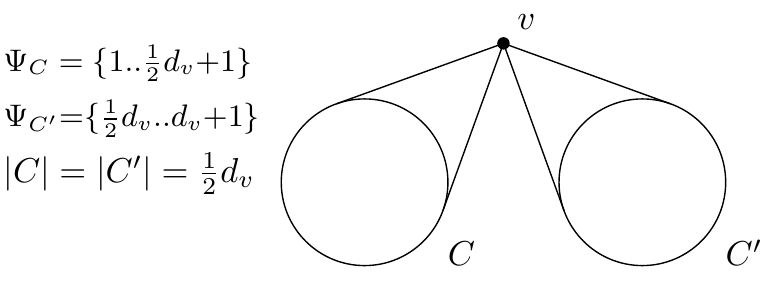}}
    \qquad
    \subfloat[\label{subfig:hard-dense}A dense node adjacent to two cliques with non-overlapping palettes.]{\includegraphics[width=.45\textwidth,page=2]{slack_issues.pdf}}
    \caption{Examples where sparsity no longer implies slack. $C$, $C'$ and $C_v$ are cliques.}
    \label{fig:hard-slack}
\end{figure}

The first and second examples in \cref{fig:hard-slack} illustrate that sparsity no longer guarantees slack in the $\degoLC$ setting. In the first example (\cref{subfig:hard-sparse}, which is from \cite{CLP20}), a sparse node $v$ is connected to two cliques with essentially non-overlapping palettes. Therefore no slack can arise from the endpoints of a non-edge in the sparse node's neighborhood picking the same color. In fact, no matter how the neighbors of $v$ get colored, it is impossible to increase the slack of $v$ from $1$ to more than $2$. The example hence shows that it can be impossible to derive hardly any slack even for sparse nodes. Thus, at least some sparse nodes need to be treated differently. We will do this by giving them temporary slack. In the example in \cref{subfig:hard-sparse}, the temporary slack is provided by coloring sparse nodes before coloring dense nodes. All the neighbors of $v$ are therefore colored after coloring $v$, giving $v$ a large amount of temporary slack. 

The second example (\cref{subfig:hard-dense}) shows that the same can also hold for dense nodes. In the $(\Delta+1)$-list coloring setting, dense nodes receive slack proportional to their external degree due to the local sparsity implied by external neighbors. This is not the case in the $(\deg+1)$-list coloring setting. The example consists of a node $v$ of high degree in a large almost-clique (making $v$ dense) and such that $v$ is adjacent to another small almost-clique. The two cliques have non-overlapping palettes as in the first example. However, here $v$ is colored as part of the dense nodes and it therefore does not automatically get temporary slack from all its dense neighbors. We handle this case by selecting a set of outliers in each almost-clique, which are handled earlier, before the remaining nodes of the clique (which we call the inliers). The inliers of a clique are nodes for which similar arguments as in the $(\Delta+1)$-case hold and we will show that a constant fraction of each almost-clique are inliers. Hence, the outliers of an almost-clique get sufficient temporary slack from the inliers, which are colored later.

\begin{figure}[th]
    \centering
    \subfloat[\label{subfig:small-deg}A node adjacent to nodes of significantly smaller degrees.]{\includegraphics[width=.45\textwidth,page=5]{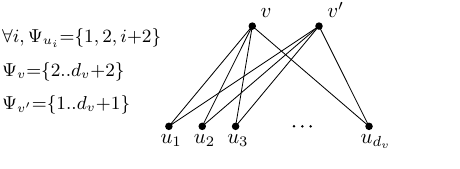}}
    \qquad
    \subfloat[\label{subfig:heavy-colors}A neighborhood in which some colors are tried $\omega(1)$ times in expectation.]{\includegraphics[width=.45\textwidth,page=4]{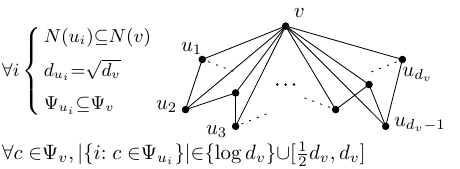}}
    \caption{Examples for which achieving concentration is hard or impossible in slack generation.}
    \label{fig:hard-concentration}
\end{figure}

The next examples in \cref{fig:hard-concentration} illustrate that even when slack exists in expectation, the usual concentration arguments might still break. The third example (\cref{subfig:small-deg}) is a case where slack exists in expectation but is impossible to achieve with concentration. Here $v$ is adjacent to nodes of much lower degree that also have another common (high-degree) neighbor.
This only happens when a node is adjacent to nodes of significantly smaller degrees, so this case disappears if we focus on coloring the nodes whose degree fall in the range $[\log^7 \Delta, \Delta]$, which we do in our main subroutine.

Finally, the fourth example (\cref{subfig:heavy-colors}) is a case where slack exists in expectation and with the probability we need, but cannot be achieved solely with same-colored pairs (as is standard for $(\Delta+1)$-(list)-coloring and is necessary for the use of martingale or Talagrand inequalities). In the example shown, neighbors of $v$ have degrees and palette sizes of order $\Theta(\sqrt{d_v})$, and $\Theta(\sqrt{d_v})$ \emph{heavy colors} appear in $\Theta(d_v)$ palettes, causing them to be tried by $\Theta(\sqrt{d_v})$ neighbors of $v$ in expectation. Other colors are unlikely to provide slack, so slack generation must rely on those heavy colors. This case is captured by our analysis for heavy colors (in \cref{sss:heavy}).

\paragraph{Additional challenges.}
The disparity in degrees and palettes in \degoLC brings numerous additional challenges that go beyond slack generation. It affects the almost-clique decomposition properties, since outside high-degree nodes can now be adjacent to even all nodes of a given almost-clique. Colors are selected with widely different probabilities and success/failure probabilities similarly vary. This impacts shattering, which is a property that depends on the maximum degree. Just the lack of knowledge of global maximum degree makes synchronization harder.

The previous state-of-the-art algorithm of \cite{CLP20} depends heavily on the global bound $\Delta$. 
The intricacy of that algorithm and its analysis is such that it is unlikely to be an effective building block for a $\degoLC$ algorithm.
The algorithm features a hierarchy of $\log\log \Delta$ decompositions that are partitioned into ``blocks'', split by size, and combined into six different sets. These are whittled down in distinct ways, resulting in three final subgraphs that are finished off by two different deterministic algorithms.
The analysis of just one of these sets runs a full 10 pages in the journal version~\cite{CLP20}.

\subsection{Algorithm Outline}

At the beginning of the algorithm, we compute an almost-clique decomposition (ACD). The ACD computation returns a partition of the nodes $V$ of the graph $G=(V,E)$ into sets $\Vsp$, $\Vun$, and into almost-cliques. Each node $v$ in $\Vsp$ is $\Omega(d_v)$-sparse (i.e., $G[N(v)]$ has $\binom{d_v}{2}-\Omega(d_v^2)$ edges), each node in $\Vun$ is $\Omega(d_v)$-uneven (i.e., $v$ has $\Omega(d_v)$ neighbors of degree $\geq(1+\Omega(1))d_v$), and in each almost-clique $C$, every $v\in C$ has at least $(1-\eacd)|C|$ neighbors in $C$ and at most $\eacd|C|$ neighbors outside $C$, for some constant $\eacd>0$. Note that the precise definitions of the ACD, sparsity, unevenness, and other related notions appear in \cref{S:prelims}. After computing the ACD, the algorithm has two main phases. We first color all the sparse and uneven nodes (i.e., all nodes in $\Vsp\cup \Vun$) and we afterwards color all the dense nodes (i.e., all the nodes $V\setminus (\Vsp\cup \Vun)$ that are in almost-cliques). In each phase, we further iterate through $O(\log^*\Delta)$ degree classes. We do this in order to be able to apply the standard shattering technique. For shattering to work, each node should succeed (in getting colored) with probability $1-1/\poly(\Delta)$. Our concentration arguments typically allow to show that each node $v$ succeeds with some probability $1-\exp(-d_v^{\Theta(1)})$ and we therefore need to make sure that when dealing with nodes of degree up to $\Delta$, the minimum node degree is at least $\log^c\Delta$ for some sufficiently large constant $c$.

\paragraph{Coloring the sparse and uneven nodes.} As observed above, unlike in the $(\Delta+1)$-list coloring problem, it is no longer true that a single round of random color trial creates sufficient slack for all sparse nodes. The high-level idea of algorithm for coloring the sparse and uneven nodes is therefore as follows. We first select a certain subset $\Vstart$ of the sparse nodes 
for which a single random coloring round might not create sufficient slack. 
Each node in $\Vstart$ has at a constant fraction of its neighbors still uncolored and outside $\Vstart$.
We then run one round of random color trial to give slack $\Omega(d_v)$ to each node $v\in \Vsp\cup\Vun\setminus \Vstart$.  With those things in place we can then color $\Vsp\cup \Vun$ as follows. In a first step, we color the nodes in $\Vstart$. Because nodes in $\Vstart$ have many neighbors outside $\Vstart$, they have temporary slack and can therefore be colored in $O(\log^* n)$ rounds by using the algorithm of \cite{SW10}. Next, we can color the remaining nodes in $\Vsp\cup \Vun$. For those we have generated enough slack in the initial random color trial step and we can therefore also color those nodes in $O(\log^* n)$ rounds by using the algorithm of \cite{SW10}.

To understand the above algorithm in more detail, we first define a set of nodes $\Veasy$ for which it is relatively easy to show that one round of random color trial creates sufficient slack. First note that this is definitely the case for all nodes in $\Vun$ and we thus have $\Vun\subseteq \Veasy$. Each node $v\in \Vun$ has $\Omega(d_v)$ neighbors $u$ of degree $d_u\geq (1+\eps)d_v$ and each such neighbor $u$ has a constant probability of choosing a color that is not in $v$'s palette. A similar argument also works more generally if $v$ has discrepancy $\Omega(d_v)$, i.e., if the average probability for $v$'s neighbors for trying a color outside $v$'s palette is constant. In this case, it is straightforward to see that the created chromatic slack is $\Omega(d_v)$ in expectation, it is however more tricky to guarantee it with sufficiently high probability (details appear in \cref{ss:disc}). All such nodes are therefore also added to $\Veasy$. Further, we call a node $v$ balanced if a large fraction of its neighbors has degree $\Omega(d_v)$. For sparse balanced nodes, essentially the same arguments as for sparse nodes in the $(\Delta+1)$-list coloring case work and the sparse balanced nodes are therefore added to $\Veasy$. Finally, nodes with a constant fraction of their neighbors that are dense also get automatic temporary slack from the fact that the dense nodes are colored after all the sparse nodes. So, these are also included in $\Veasy$.

An additional class of nodes $v$ that we can prove 
obtain slack $\Omega(d_v)$ are nodes for which a constant fraction of the neighbors is expected to try a color that is 'heavy'. Here, a color is called heavy if the expected number of neighbors of $v$ trying this color is at least a sufficiently large constant. We call those nodes $\Vheavy$. For nodes in $\Vheavy$, it is straightforward to see that the expected slack from neighbors picking the same color is $\Omega(d_v)$. However, in this case we have to invest some additional work to prove that this slack is also created with a sufficiently large probability (see \cref{sss:heavy}). We can now define the set $\Vstart$ as follows. $\Vstart$ contains all nodes $v\in \Vsp\setminus (\Veasy\cup\Vheavy)$ such that a constant fraction of the neighbors of $v$ are 
in $\Veasy$. The final set of nodes that are not classified are the nodes in $\Vsp\setminus (\Veasy\cup\Vheavy\cup\Vstart)$. We call those nodes tough and in \cref{ss:tough} we show that tough nodes also obtain sufficient slack in the initial round of random color trial.

\paragraph{Coloring the dense nodes.} As a first step, each almost-clique $C$ defines a leader node $x_C\in C$ and a set of outlier nodes $O_C\subseteq C$. The leader of $C$ is the node $x_C\in C$ of minimum \emph{slackability}, where slackability is defined as the sum of discrepancy and sparsity. The slackability of $x_C$ will also be referred to as the slackability of the almost-clique $C$. The set of outliers $O_C$ consists of the (approximately) third of the nodes in $C$ with the fewest common neighbors with $x_C$, of the sixth of the nodes in $C$ of maximum degree, and of the antineighbors of $x_C$ in $C$ (i.e., the nodes that are not adjacent to $x_C$). The remaining nodes of $C$ (which is still close to at least half of $C$) is called the inliers of $C$. We show that all the inliers $I_C=C\setminus O_C$ of a clique $C$ have similar properties (and in particular neighborhoods and palettes that are near-identical, with differences on the order of the slackability of $x_C$).

After defining the leader and outliers of each almost-clique, we run one round of random color trial to create slack. For each almost-clique $C$ with slackability at least $\log^c\Delta$ for a sufficiently large constant $c$, we show that all the inliers obtain slack that is proportional to their slackability (and thus in particular at least proportional to the slackability of the almost-clique). The arguments for slack generation are similar to the corresponding arguments for sparse nodes (details in \cref{S:slackgen}). 

After slack generation, we select one more set in each almost-clique $C$. For each almost-clique $C$, we compute a 'put-aside' set $P_C$ as follows. We first choose a random subset $S_C$ of $I_C$ of size $\poly\log\Delta$, inducing a global set $S = \cup_C S_C$. To obtain $P_C$, we then remove any node from $S_C$ with a neighbor in $S$. 
Note that the sets $P_C$ of different almost-cliques are independent and they can therefore be colored trivially even if all other nodes are already colored. We can therefore delay coloring those sets to the very end of the algorithm. With sufficiently high probability, the set $P_C$ of each almost-clique $C$ is of sufficiently large polylogarithmic size. 
We need the sets $P_C$ to create temporary slack for the other nodes in ultradense almost-cliques in which the slack generation is a low-probability event.

We can now proceed to color most of the nodes of the almost-cliques. In a first step, we color all the outliers. Because the outliers are only roughly at most half of each almost-clique, they have sufficient slack from the inliers so that they can be colored in $O(\log^* n)$ rounds by using the algorithm of \cite{SW10}. After coloring the outliers, we color most of the inliers of each clique. Here, we use the fact that the leader of each clique is connected to all the inliers of the clique and that the leader's color palette is not too different from the color palettes of the other inliers. The leader $x_C$ therefore just randomly proposes one of its own available colors to each of the nodes in $I_C$, so that no color is proposed more than once. It is remarkable that this simple primitive suffices to color nearly all the inliers, leaving only a portion proportional to the slackability of $C$. 
The remaining inliers then have slack proportional to their remaining degree (where the slack in  ultradense almost-cliques comes from the put-aside set $P_C$). We can therefore fully color them with the algorithm of \cite{SW10}. At the very end, we finally color the nodes in the put-aside sets $P_C$.

\paragraph{Putting everything together.}
The combination of our algorithm for sparse and uneven nodes and our algorithm for dense nodes gives us an algorithm to color all nodes of degree $\Omega(\log^c n)$ in $O(\log^* n)$ rounds, w.h.p. Applied to nodes in lower degree range, the combined algorithm shatters the subgraph associated to the degree range in $O(\log^*n)$ rounds. We apply the combined algorithm to the subgraphs induced by $O(\log^*\Delta)$ degree classes, starting from the higher degrees. Each time, we color the shattered graph with a deterministic algorithm whose running time decreases as the maximum degree of the graph goes down. This decreasing cost of the deterministic algorithm means that the running time is dominated by the cost of the deterministic algorithm applied to the second degree range, consisting of nodes of degree $\Omega(\poly(\log\log n)) \cap O(\poly(\log n))$. In combination with the $O(\log^2 \Delta\cdot\log n)$-round deterministic $(\deg+1)$-list coloring algorithm of \cite{GK21}, this leads to an overall round complexity of $O(\log^3\log n)$.

\section{Preliminaries and Definitions}
\label{S:prelims}

\paragraph*{Constants and evolving quantities.} Throughout the paper, we use subscripts for constant numerical quantities and parentheses for evolving ones, e.g., $d_v$ and $\pal_v$ are the original degree and palette of node $v$, while $d(v)$ and $\pal(v)$ are the current degree and palette, i.e., taking into account that parts of the graph have been colored or turned off.

Let us consider $\Delta$ as an upper-bound on the maximum degree rather than the maximum degree itself. 
Let $\ell = \log^{2.1} \Delta$.

\subsection{Slack, Sparsity, \& Almost-Cliques}

\begin{definition}[Slack]
\label{def:slack}
The \emph{slack} $s(v)$ of a node $v$ in a given round is the difference $\card{\pal(v)} - d(v)$ between the number of colors it has then available and its degree in that round.
\end{definition}

For any subset of the vertices $S \subseteq V$, we denote by $E[S] = E\cap \binom{S}{2}$ the set of edges between nodes of $S$, and by $m(S) = \card{E[S]}$ the number of edges between nodes of $S$. The next quantity (sparsity) measures \emph{the number of missing edges} in a node's neighborhood. Note that the definition used here is different from the one used when dealing with \DeloC or \DeloLC, to address the variability of the palette sizes.

\begin{definition}[Sparsity]
\label{def:sparsity}
The \emph{(local) sparsity} $\zeta_v$ of node $v$ is defined as $\frac{1}{d_v}\cdot\left[\binom{d_v}{2}-m(N(v))\right]$. Node $v$ is \emph{$\zeta$-sparse} if $\zeta_v\ge \zeta$, and \emph{$\zeta$-dense} if $\zeta_v\le \zeta$.
\end{definition}
To address the variety in size and content of the palettes that are inherent to \degoLC, we use several quantities that measure how much a node's palette differs from its neighbors'.

\begin{definition}[Disparity, Discrepancy \& Unevenness]
\label{def:unevennes}
The \emph{disparity} of $u$ towards $v$ is defined as $\disc_{u,v} = \card{\pal_u \setminus \pal_v} / \card{\pal_u}$. The \emph{discrepancy} of node $v$ is defined as $\disc_v=\sum_{u \in N(v)} \disc_{u,v}$, and
its \emph{unevenness} is defined as $\unev_v=\sum_{u \in N(v)} \frac{\max(0,d_u - d_v)}{d_u+1}$. Node $v$ is \emph{$\disc$-discrepant} if $\disc_v\ge \disc$, \emph{$\unev$-uneven} if $\unev_v\ge \unev$.
\end{definition}
It always holds that $\disc_v \ge \unev_v$, and the two are equivalent in the non-list $\deg+1$ setting. In addition to the fixed quantities defined here, we also make use of the evolving variant $\disc(u,v) = \card{\pal(u) \setminus \pal(v)} / \card{\pal(u)}$ later in the paper. Intuitively, discrepancy is how many neighbors of a node are expected to try a color outside its palette, and disparity is the contribution of individual nodes to that quantity. Unevenness focuses on how much the palettes differ in size, ignoring their content.

Sparsity and (more recently) unevenness have been key in the definition of graph decompositions known as \emph{almost-clique decompositions}. Intuitively, such decompositions partition the graph into small-diameter connected components of dense and even nodes on the one hand and possibly big sets of comparatively sparse or uneven nodes on the other hand. We use an almost-clique decomposition of \cite{AA20}, tailored to the $\deg+1$ setting. See also earlier $\Delta+1$-oriented ACD definitions of \cite{HSS18,ACK19}.
\begin{definition}[($\deg+1$) ACD \cite{AA20}] \label{def:acd}
Let  $G=(V,E)$ be a graph and $\eacd,\espa\in (0,1)$ be parameters. A partition $V=\Vsp \sqcup \Vun \sqcup \Vdense$ of $V$, with $\Vdense$ further partitioned into $\Vdense = \bigsqcup_{C \in \acset} C$,
is an \emph{almost-clique decomposition (ACD)} for $G$ if: \begin{compactenum}
    \item Every $v \in \Vsp$ is $\espa d_v$-sparse\ ,
    \item Every $v \in \Vun$ is $\espa d_v$-uneven\ ,
    \item For every $C \in \acset$ and $v\in C$, $d_v \leq (1+\eacd)\card{C}$\ ,
    \item For every $C \in \acset$ and $v\in C$, $(1+\eacd)|N_{C}(v)|\ge \card{C}$\ .
\end{compactenum}
\end{definition}
As is shown in \cite{AA20}, 
An ACD  can be found in a constant number of rounds in \LOCAL \cite{AA20}, for any $\eacd > 0$ and $\espa=\Omega(\eacd^2)$.
 We refer to the $C$'s as \emph{almost-cliques}. For each $C \in \acset$ let $\Delta_C = \max_{v\in C}d_v$, and for each $v \in \Vdense$ let $C_v$ be the almost-clique containing $v$. Properties 3 and 4 of \cref{def:acd} directly imply that for every $C \in \acset$, $(1-\eacd)\Delta_C\le\card{C}\le(1+\eacd)\Delta_C$, 
 and that for every $v \in \Vdense$, $d_v\geq(1-2\eacd)\Delta_{C_v}$. 
 It also follows that the diameter of each $G[C]$ is at most 2.

Almost-clique decompositions anterior to \cite{AA20} were tailored to solve $\Delta+1$-coloring problems. As such, they used a definition of sparsity involving the maximum degree $\Delta$ of $G$, had no notion of unevenness, and did not consider almost-cliques of size $o(\Delta)$. Such ACDs could be found for any graph in a constant number of rounds of {\LOCAL} \cite{HSS18} or {\CONGEST} \cite{HKMT21}. The type of decomposition presented here, tailored to $\deg+1$-coloring problems and due to Alon and Assadi~\cite{AA20}, can similarly be computed in constant rounds of \LOCAL.

In the $\Delta+1$ setting, a simple link exists between sparsity and slack: a simple randomized procedure gives slack to nodes that have sparsity. In this setting sparsity is also useful in analyzing the structural properties of almost-cliques. The situation is very different in the $\deg+1$ setting, as will be evident from our analysis of slack generation in this paper. Notably, sparsity alone is no longer sufficient as a quantity for slack generation and the structural analysis of almost-cliques, leading to our introducing \emph{slackability}.
\begin{definition}[Slackability]
\label{def:slackability}
The \emph{slackability} $\barsigma_v$ of node $v$ is defined as $\barsigma_v = \disc_v+\zeta_v$.
We also define the \emph{strong slackability} as $\sigma_v = \unev_v+\zeta_v$.
\end{definition}

Schneider and Wattenhofer \cite{SW10} showed that coloring can be achieved ultrafast if all nodes have slack at least proportional to their degree (and the degree is large enough). This is achieved by each node trying up to $\log n$ colors in a round, using the high bandwidth of the \LOCAL model. We use the following variant that is very similar but still slightly different from some previous results. For instance, the case where $\kappa=1$ is a direct consequence of Lemma 2.1 in~\cite{CLP20}.

\begin{restatable}{lemma}{slackcolorlemma}
\label{lem:slackcolor}
Consider the $\deg+1$-list coloring problem where each node $v$ has slack $s(v)=\Omega(d(v))$.
Let $1<\smin \leq \min_v s(v)$ be globally known.
For every $\kappa\in (1/\smin,1]$, there is a randomized \LOCAL algorithm {\slackcolor[$(\smin)$]} that in $O(\log^* \smin+1/\kappa)$ rounds properly colors each node $v$ w.p.\ $1 - \exp(-\Omega(\smin^{1/(1+\kappa)})) - \Delta e^{-\Omega(\smin)}$, even conditioned on arbitrary random choices of nodes at distance $\geq 2$ from~$v$.
\end{restatable}
 We give a proof of \cref{lem:slackcolor} and a description of \slackcolor in \cref{app:slackcolor} for completeness.
 
\subsection{Basic Primitive}

The basic primitive in randomized coloring algorithms, which we call {\tryrandomcolor}, is for nodes to \emph{try} a random eligible color: propose it to its neighbors and keep it if it does not conflict with them. More formally, we run {\trycolor} (\cref{alg:trycolor}), with an independently and uniformly sampled color $\col_v\in \pal(v)$.
A more refined version gives priority to some nodes over others: for each node $v$, we partition its neighborhood $N(v)$ into $N^+(v)$ -- the nodes whose colors conflict with $v$'s -- and $N^-(v) = N(v)\setminus N^+(v)$. For correctness of \trycolor, $u \in N^-(v) \rightarrow v\in N^+(u)$ should hold for each edge $uv$. The standard algorithm, where all nodes conflict with each other, corresponds to setting $N^+(v)=N(v)$, for all $v$.
Repeating it leads to a simple $O(\log n)$-round algorithm \cite{johansson99}.
\begin{algorithm}[H]\caption{\trycolor (vertex $v$, color $\col_v$)}\label{alg:trycolor}
\begin{algorithmic}[1]
\STATE Send $\col_v$ to $N(v)$, receive the set $T^+=\{\col_u : u\in N^+(v)\}$.
\STATE{\textbf{if}} $\col_v\notin T^+$ \textbf{then} permanently color $v$ with $\col_v$.
\STATE Send/receive permanent colors, and remove the received ones from $\pal(v)$.
\end{algorithmic}
\end{algorithm}

\begin{algorithm}[H]\caption{\tryrandomcolor (vertex $v$)}\label{alg:tryrandomcolor}
\begin{algorithmic}[1]
\STATE Pick $\col_v$ u.a.r.\ from $\pal_v$. 
\STATE $\trycolor(v,\col_v)$
\end{algorithmic}
\end{algorithm}

\section{Coloring Sparse and Uneven Nodes}
\label{S:sparse}

It is well established \cite{SW10,EPS15} that if nodes have slack proportional to their degree, then they can be colored ultrafast ($O(\log^* n)$ time for high-degree nodes) by \slackcolor. Sparse nodes have sparsity linear in their degree. This leads to linear slack in the $\Delta+1$-coloring problem, using the following simple algorithm \slackgeneration.
 
 \begin{algorithm}[H]\caption{\slackgeneration[(probability $\pgen$)]}\label{alg:slackgeneration}
\begin{algorithmic}[1]
\STATE $S\gets $ sample each $v\in G$ into $S$ independently w.p.\ $\pgen=1/10$.
\STATE \algorithmicforall\ $v\in S$ in parallel \algorithmicdo\ {\tryrandomcolor}$(v)$.
\end{algorithmic}
\end{algorithm}

We also use \slackgeneration for \degoLC, but as we have seen, this is not sufficient to generate slack for all nodes.
Our solution is to identify a particular subset of sparse nodes, $\Vstart$ (to be detailed shortly) that don't get slack in the classical way. We then show these nodes can still be colored fast \emph{if} they are colored before the other sparse nodes, $\Vsp\setminus \Vstart$. This is formalized in the following lemma.

\begin{restatable}{proposition}{slacklemmasparse}\label{L:slackgen-sparse}
Assume all nodes have degree at least $s \geq C\cdot \ln^2 \Delta$ for some universal constant $C$. There is a $O(1)$-round procedure that identifies a subset $\Vst\subseteq \Vsp$ such that after running {\slackgeneration} in the subgraph induced by $\Vsp \cup \Vun$:
\begin{compactenum}
    \item Each node $v$ in $\Vst$ has $\Omega(d_v)$ uncolored neighbors in $V \setminus \Vst$ w.p.\ $1 - \exp(-\Omega(d_v))$, and
    \item Each node $v$ in $\Vun \cup \Vsp\setminus \Vst$ has slack $\Omega(d_v)$, 
    w.p.\ $1 - \exp(-\Omega(\sqrt{s}))$.
\end{compactenum}
For each node, the probability bounds hold even when conditioned on arbitrary random choices outside its 2-hop neighborhood.
\end{restatable}

The proof of \cref{L:slackgen-sparse} appears in \cref{ss:tough}.
Assuming \cref{L:slackgen-sparse}, we have the following simple procedure for coloring sparse nodes.

\begin{algorithm}[H]
\caption{Main Procedure for Coloring Sparse Nodes}
\label{alg:logstar-sparse}
    \begin{algorithmic}[1]
    \STATE Identify the set $\Vst \subset \Vsp$
    \STATE {\slackgeneration} in $G[\Vsp\cup \Vun]$.
    \STATE {\slackcolor} $\Vst$. \label{st:sp-multitrial}
    \STATE {\slackcolor} $\Vsp \setminus \Vst$ and $\Vun$. \label{st:o-multitrial2}
\end{algorithmic}
\end{algorithm}
 
We now describe the set $\Vstart$, along with informal versions of all the relevant definitions. We then sketch the arguments used in proving the slack generation result, including the distinct cases treated.
We defer proof details to \cref{S:slackgen}.
We define and use a number of small epsilon constants in the formal definitions. For reference, here are their order of magnitude in relation to $\eacd$: $\espa,\eub \in \Theta(\eacd^2)$; $\ehat \in \Theta(\eacd^4)$; $\ehc \in \Theta(\eacd^8)$.

A sparse node $v$ is said to be \emph{balanced} if most of its neighbors are of degree at least $2d_v/3$: $|\{u\in N(v): d_u \ge 2d_v/3\}| \ge (1-\eub) d_v$.
A node is \emph{discrepant} if its discrepancy is at least a constant fraction of its degree: $\disc_v \ge \ehat d_v$.
This case subsumes the uneven case,
in which a node has a constant fraction of its neighbors with a non-trivially larger degree. 
The \emph{easy} nodes $\Veasy$ are the uneven nodes and the sparse nodes that are either balanced, discrepant, or with dense nodes making up a constant fraction of their neighborhood. These obtain slack with standard arguments.

Another class of nodes that receives permanent slack from \slackgeneration are the heavy nodes, defined informally as follows. 
  The weight of a color $c$ equals the expected number of neighbors of $v$ that pick that color in {\slackgeneration}: $H(c)=H_v(c)=\sum_{u\in N(v), c\in \pal_u}\frac{1}{|\pal_u|}$. Let $\hcol = \hcol_v = \{c \in \pal_v : H(c) \ge \ehc\}$ be the set of \emph{heavy colors} for $v$. 
A node is \emph{heavy} if the total weight of its heavy colors is a constant fraction $\ehat$ of its degree: $\sum_{c \in \hcol_v} H(c) \ge \ehat d_v$

We can now define $\Vstart$, the nodes that should be colored first. Those are the sparse nodes that are not heavy nor easy, but have a constant fraction $\ehat$ of their neighbors that are easy. These easy neighbors therefore provide temporary slack for the node, if it is colored before them.

Formally, we define the following sets of nodes:
\[ \begin{array}{lll}
    \Vbal & = \{v \in \Vsp: |\{u\in N(v): d_u \ge 2d_v/3\}| \ge (1-\eub) d_v \}  & \text{(balanced nodes)} \\
    \Vdisc  & = \{v \in \Vsp: \disc_v \ge \ehat d_v \}  &  \text{(discrepant nodes)} \\
    \Veasy & = \Vbal \cup \Vdisc \cup \Vun \cup 
    \{ v \in \Vsp : |N(v)\cap \Vdense| \ge \ehat d_v \} &  \text{(easy nodes)} \\
    \Vheavy & = \{v \in \Vsp\setminus \Veasy : \sum_{c \in \hcol_v} H(c) \ge \ehat d_v \}  &  \text{(heavy nodes)} \\
    \Vst & = \{ v \in \Vsp\setminus (\Veasy\cup\Vheavy) : |N(v)\cap \Veasy| \ge \ehat d_v \}  & \text{(nodes colored first)} \\
    \Vtough & = \Vsp \setminus (\Veasy \cup \Vheavy \cup \Vst) & \text{(tough nodes)}
\end{array} \]

\paragraph{Proof intuition.}
As mentioned, standard arguments suffice to show that easy nodes ($\Veasy$) get slack. Also, it is immediate that the nodes of $\Vst$ get temporary reprieve from their waiting neighbors.
The remaining sparse nodes fall into two types.

There are the heavy nodes (specifically those that are not easy), which have many ``heavy colors'' in their neighborhood.
Each heavy color can contribute a large amount of slack in expectation, and a change in the color of a single node can decrease the expected total contribution of other nodes significantly. Thus, the usual concentration bounds do not apply. 

We tackle this by a two-stage analysis. We show that there exists a partition of the colorspace into buckets with some nice properties and fix one such partition (only for the sake of the analysis). We view the random color choice as consisting of two steps: picking a bucket, and picking a color within that bucket. We can derive tight bounds on the number of nodes and the number of their neighbors that select a given bucket. We can then analyze each bucket in isolation, for which it suffices to obtain bounds on the expected number of nodes colored with each heavy color.
We can then use Hoeffding bound to get a concentration lower bound on the total number of nodes colored with heavy colors. This bound is significantly larger than the number of heavy colors, which implies that w.h.p.\ many colors are reused, i.e., linear slack is generated.

The remaining sparse nodes that fall into none of the types above (i.e., they are light and neither in $\Veasy$ nor $\Vst$) are said to be \emph{tough}. 
One of the main result is that the tough nodes do get permanent slack from \slackgeneration (\cref{alg:slackgeneration}).
At a high level, we orient the edges from high to low degree and sum the in- and out-degrees of the neighbors of a tough node. A gap exists between the sums due to the large number of unbalanced neighbors, which implies the presence of slack-providing non-edges.
The finer details for this are not very easily intuitive, and we defer the discussion to the detailed presentation in \cref{S:slackgen}.

\section{Coloring Dense Nodes}
\label{S:dense}

We give now an algorithm for graph $H$ containing only dense nodes.
Once the sparse (and uneven) nodes have been colored, we are indeed left with a graph consisting only of dense nodes, so we can view $H$ as the subgraph induced by $\Vdense$. In the original graph $G$, at most an $\eacd$-fraction of each dense node's neighborhood is non-dense, so their degrees in $H$ are all at least their original degree times $(1-\eacd)$ and fall into essentially the same degree range. Observe that an almost-clique decomposition of $G$ is still a valid decomposition of $H = G[\Vdense]$, as conditions 3 and 4 of \cref{def:acd} remain satisfied. (The opposite is not true: after coloring the dense nodes, the sparse nodes may no longer be sparse.) We are in a sense using the self-reducibility property of the \degoLC. 

The algorithm (\cref{alg:logstar-dense}) builds on previous frameworks for randomized coloring (\cite{HSS18,CLP20}), but with several notable changes. Some of the notable differences from some or most previous approaches include:
\begin{compactenum}
    \item Management of palette discrepancy (both in size and color composition), by separately treating those with the largest variance; 
    \item A procedure that generates slack to each dense node proportional to its sparsity;
    \item A procedure to give temporary slack to nodes within very isolated almost-cliques, for which the previous argument provides little slack or with insufficient probability; and
    \item A single-round procedure to color most nodes in an almost-clique by synchronizing the colors they try.
\end{compactenum}

Recall that $\ell = \log^{2.1} \Delta$.
We say that $C$ is a \emph{low-slack} almost-clique if $\barsigma_C \le \ell$.
Let $\barsigma_C$ be the minimum $\barsigma_v$ over nodes $v$ in $C$.
Please note that definitions of dense nodes, such as slackability, are in terms of $H = G[\Vdense]$, i.e., the subgraph induced by $\Vdense$.

\begin{algorithm}[H]
\caption{Main Procedure for Coloring Dense Nodes}
 \label{alg:logstar-dense}
  \begin{algorithmic}[1]
  \STATE Compute the leader $x_C$ and outliers $O_C$ of each almost-clique $C$. Let $O = \cup_C O_C$. \label{st:outliers-hi}
  \STATE {\slackgeneration}.
  \STATE $P_C \gets \putaside[(C)]$ in each low-slack almost-clique $C$. Let $P = \cup_C P_C$.      \label{st:putaside-hi}
  \STATE {\slackcolor} $O$. \label{st:o-multitrial}
    \STATE {\synchronizedcolortrial} $\Vdense \setminus P$.\label{st:synchtrial}
    \STATE {\slackcolor} $\Vdense \setminus P$. \label{st:lastmultitrial-hi}
    \STATE For each low-slack $C$, let $x_C$ collect the palettes in $P_C$ and color the nodes locally.\label{st:collect-hi}
\end{algorithmic}
\end{algorithm}

We first derive structural bounds on dense nodes in  \cref{ss:slackabilitybounds}.
We then treat the steps 1, 2, 3 and 5 of the algorithm in individual subsections.

\subsection{Slackability Bounds External and Anti-Degree}
\label{ss:slackabilitybounds}

\begin{definition}[External/anti-degree] \label{def:ext-anti-degree}
For a node $v \in V \setminus (\Vsp \cup \Vun)$, let $C_v$ denote its almost-clique,
$E_v = N(v) \setminus (C_v \cup \Vsp \cup \Vun)$ its set of \emph{external neighbors}, and $e_v = |E_v|$ its \emph{external degree}.
Similarly, let $A_v = C_v \setminus N(v)$ denote its set of \emph{anti-neighbors} and $a_v = |A_v|$ its \emph{anti-degree}.
\end{definition}

In the $\Delta+1$ setting, it was recently observed \cite{HKMT21} that the sparsity of a node bounds its external and anti-degrees. As sparsity implies that a proportional amount of slack can be (probabilistically) obtained in this setting, this meant that nodes could be guaranteed to have external and anti-degree bounded by their slack.
We show an analogous result here where strong slackability $\sigma_v = \unev_v + \zeta_v$ replaces sparsity.

\begin{lemma}
\label{lem:low-ext-degree}
There is a constant $\cext=\cext(\eacd)$
such that $e_v\le \cext\cdot \sigma_v$ holds for every node $v$ in an almost-clique $C$. 
\end{lemma}
\begin{proof}
Let $u\in E_v$ be an external neighbor of $v$, i.e., $u$ is a neighbor of $v$ in an almost-clique $C' \neq C$. Nodes $u$ and $v$ are mostly adjacent to other nodes of their almost-cliques: $\card{N_{C}(v)} \geq (1-2\eacd)d_v$ and $\card{N_{C'}(u)} \geq (1-2\eacd)d_u$, and therefore, $\card{N(v) \cap N(u)} \leq 2 \eacd (d_v + d_u)$.

This immediately implies that each such $u$ contributes $\Omega(1)$ to $v$'s strong slackability: if $d_u \leq 2d_v$, then $u$ is part of at least $(1-6\eacd)d_v$ non-edges in $v$'s neighborhood, and thus contributes $\Omega(1)$ to $\zeta_v$; otherwise, $d_u > 2d_v$ has $u$ contribute $\Omega(1)$ to $\unev_v$.
\end{proof}

\begin{lemma}
\label{lem:low-anti-degree}
There is a constant $\cant=\cant(\eacd)$
such that $a_v\le \cant\cdot \sigma_v$ holds for any dense node $v$.
\end{lemma}
\begin{proof}
Let $C=C_v$. We bound the unevenness via the degree sum of the nodes in $N_C(v) = N(v)\cap C$:
\begin{equation}    \sum_{u \in N_C(v)} d_u \ge \sum_{u\in N_C(v)} d_u[N(v)] + \sum_{u\in N_C(v)} d_u[A_v]\ ,
\label{eq:degsumbound}
\end{equation}
where, for a set $S$, we let $d_u[S]=|N(u)\cap S|$.
There are only $\zeta_v d_v$ edges missing within $N(v)$,   thus the first degree sum on the right-hand side above ``misses'' only the corresponding at most $2\zeta_v d_v\le 2(1+\eacd)\zeta_v|C|$ ``half-edges'', that is,
\[ \sum_{u\in N_C(v)} d_u[ N(v)] \ge  \sum_{u\in N_C(v)} d_u - 2(1+\eacd)\zeta_v |C|\ . \]
To bound the second sum, let us rearrange it as a sum over $A_v$, and recall that each node in $A_v$ has at least $(1-\eacd)|C|$ neighbors in $C_v$, and $|A_v|\leq \eacd |C|$  
(by the ACD property):
\[ \sum_{u\in N_C(v)} d_u[A_v] =\sum_{w\in A_v} d_w[N_C(v)]\ge \sum_{w\in A_v} d_w[C] - d_w[A_v] \ge |A_v| ((1-\eacd)|C| - \card{A_v}) \ge (1-2\eacd) a_v |C|\ . \]
Plugging these bounds back into (\ref{eq:degsumbound}), rearranging, and dividing by $|C|$, we get:
\[ \frac{\sum_{u \in N_C(v)} (d_u - d_v)}{|C|} \ge (1-2\eacd)a_v - 2(1+\eacd)\zeta_v \ . \]
Thus, since $d_u\le (1+\eacd)|C|$ holds for $u\in C$,
\[ \unev_v \ge \sum_{u \in N_C(v)} \frac{d_u - d_v}{d_u} \ge \sum_{u \in N_C(v)} \frac{d_u - d_v}{(1+\eacd)|C|}\ge (1-3\eacd)a_v - 2\zeta_v \ . \]
Hence, $a_v \le 2/(1-3\eacd) \cdot \sigma_v$.
\end{proof}

\subsection{Selecting the Leaders and Inliers}
\label{S:inliers}

An important property of almost-cliques used in recent randomized algorithms for coloring (\cite{HKM20,HKMN20}) is their relative uniformity when it comes to certain quantities (notably sparsity and external degree). In a natural continuation of previous works, we show that any almost-clique contains a constant fraction of nodes (the \emph{inliers}) with the same slackability $\barsigma_C$ up to a multiplicative constant, and the same degrees up to an additive $O(\barsigma_C)$ term. By taking these inliers w.r.t.\ a well chosen \emph{leader} we can additionally ensure that inliers' palettes significantly overlap with that of the leader. These properties are key to the success of {\synchronizedcolortrial} in Step~\ref{st:synchtrial} of the algorithm.

We choose as leader in $C$ the node $x = x_C$ of minimum slackability $\barsigma_x = \barsigma_C$. We eliminate:
\begin{compactenum}
    \item the $\max(d_x,\card{C})/3$ nodes in $C$ with the fewest common neighbors with $x$,
    \item the $\card{C}/6$ nodes of largest (original) degree, and 
    \item the anti-neighbors $A_x$ of $x$.    
\end{compactenum}

The eliminated nodes in $C$ form its set $O_C$ of \emph{outliers}. The remaining set $\core_C = C \setminus O_C$ of \emph{inliers} is of size at least $|C|(1/2 - 2\eacd) = \Omega(|C|)$. This size implies that coloring the outliers $O=\bigcup_{C\in \acset} O_C$ before the inliers gives them a large amount of temporary slack. Observe that the leader $x$ might become an outlier, but it continues to fulfil its role.

The following lemma shows that the nodes in $\core$ share most of their neighborhoods and palettes and have close to uniform degrees, even when part of the graph is colored. Recall that $\triangle$ denotes the symmetric difference of two sets (not to be confused with the Greek letter $\Delta$), and that $\pal_u$ ($\pal(u)$) refers to the original (evolving) palette of $u$. 

\begin{lemma}
For all $u \in \core_C$, it always holds that $e_u\le 12\sigma_x$, 
$\card{N(u) \triangle N(x)} \leq 12 \sigma_x$ and $\card{\pal(u) \triangle \pal(x)} \leq \card{\pal_u \triangle \pal_x} + 12 \sigma_x$.
\label{lem:degreesinac}
\end{lemma}
\begin{proof}
By definition of $\zeta_x$, there are at most $\zeta_x d_x$ missing edges in $N(x)$, therefore at most $d_x/3$ neighbors of $x$ are non-neighbors of more than $6\zeta_x$ other neighbors of $x$.
Since such nodes were eliminated (due to the first criteria for inclusion in $O_C$), the nodes in $\core_C$ all have at least $d_x - 6\zeta_x \ge d_x - 6\sigma_x$ common neighbors with $x$.

Let $Q$ be the set of nodes of $N_C(x)$ of degree at least $d_x + 5 \unev_x$. 
By definition of unevenness,
\[ \unev_x \ge \sum_{w\in Q} \frac {d_w-d_x}{d_w+1} 
\ge \frac{5 \unev_x |Q|}{(1+\eacd)\card{C}}\ . \]
Thus, $|Q| \le (1+\eacd)|C|/5 \le |C|/6$. Since the $|C|/6$ nodes of highest degree in $C$ were eliminated (due to the second criteria for $O_C$), all nodes in $\core$ have (initial) degree at most $d_x + 5\unev_x \le d_x + 5\sigma_x$.

Given the large common neighborhood within $C$ and the degree bounds, the size of the symmetric difference $N(u) \triangle N(x)$ is at most $2\cdot 6\zeta_x + 5\unev_x \le 12\sigma_x$, and the same holds for
the external degree of $u$. This bound holds as nodes get colored since those quantities can only get smaller. As corollary, $d(u)$ and $d(x)$ never differ by more than $12 \sigma_x$. The symmetric difference $\pal(u) \triangle \pal(x)$ only evolves either when a node in $N(u) \triangle N(x)$ gets colored, or when a color in $\pal(u) \triangle \pal(x)$ gets picked by a common neighbor of $x$ and $u$. The first can only happen $12 \sigma_x$ times by the bound on $N(u)\triangle N(x)$, and the second only decreases the symmetric difference, hence the claim.
\end{proof}

\subsection{Slack Generation for Dense Nodes}
\label{ss:slackgeneration-dense}

Our result on slack generation for dense nodes can be summarized by the following proposition.

\begin{restatable}{proposition}{slacklemmadense}\label{L:slackgen-dense}
There is a $O(1)$-round \CONGEST algorithm {\slackgeneration} that colors a subset of the dense nodes such that afterwards:
\begin{compactenum}
    \item Each node $v$ in $\Vdense \setminus O$ of slackability $\barsigma_v\geq \ell$ has slack $\Omega(\barsigma_v)$, w.p.\ $1 - \exp(-\Omega(\ell))$. \label{I:dense-slack}
    \item Each node $v$ in $V$ has at most $d_v/4$ colored neighbors, w.p.\ $1 - \exp(-\Omega(d_v))$.
\end{compactenum}
For each node, the probability bounds hold even conditioned on arbitrary random choices outside its 2-hop neighborhood.
\end{restatable}

The proof of \cref{L:slackgen-dense} appears in \cref{ss:slackgen-dense}.
Part 2 follows from a straightforward application of Chernoff bounds. Part 1 is achieved by treating different types of nodes and
analyzing the slack from sparsity and the slack from discrepancy separately. 

The slack (both sparsity- and discrepancy-based) obtained from neighbors of comparable or higher degree can be bounded via standard approaches (as in the sparse case). The impact of heavy colors, as well as the discrepancy from light colors, both follow from the same arguments as used in the sparse case.
The hardest part is in bounding the slack from sparsity involving light neighbors of lower degree. In particular, the main effort is spent on the \emph{gritty} nodes that are light, non-discrepant, unbalanced, and not among the outliers. 
What makes the dense case though easier than the sparse case is that 
it suffices to consider only non-edges with one endpoint in $v$'s almost-clique $C_v$ and the other in an external neighbor $w$ of $v$.
Each neighbor $w$ of $v$ of lesser degree has at least $\eacd |C_v|$ non-neighbors in $v$. If those pairs of nodes have very different palettes, then we argue that we get slack via discrepancy; otherwise, they lead to slack via the standard sparsity argument.

\subsection{Put-aside Sets for Low-Slack Almost-Cliques}
\label{ss:putaside}

While low-slack almost-cliques in expectation obtain the slack and degree reduction needed to apply {\slackcolor}, they do not obtain them with a sufficiently high probability. Even if it were the case, {\slackcolor} itself would not work with a sufficiently high probability. Fortunately, their low slackability implies that they are poorly connected to the rest of the graph, which allows us to put aside a subset of their nodes such that: 1/ it gives enough temporary slack for {\slackcolor} to color the rest of the almost-clique efficiently; 2/ the put-aside sets are easily colored when the rest of the graph has been dealt with.

Consider a low-slack almost-clique $C$. Recall that this means that $\barsigma_C \leq \ell$. By definition, $\ell$ is s.t.\ $\Delta_C \in \Omega(\ell^3)$ 
since $\Delta_C = \Omega(\log^7 \Delta)$. 
Put-aside sets are computed according to \cref{alg:pas}.

\begin{algorithm}[H]\caption{{\putaside}$(C)$} 
\label{alg:pas}
  \begin{algorithmic}[1]
  \STATE $S_C\gets$ each node $v\in \core_C$ is sampled independently w.p.\ $\pdisj = \ell^2/(48\Delta_C)$.
  \RETURN $P_C \gets \{v\in S_C : E_v\cap S = \emptyset\}$, where $S = \cup_{C'} S_{C'}$
 \end{algorithmic}
\end{algorithm}

Intuitively, we use that each inlier has $O(\ell)$ external neighbors due to the manner external degree and slackability are related (\cref{lem:low-ext-degree}), so expected at most $O(\ell\cdot \ell^2/\Delta_C) = O(1)$ of them get sampled.

\begin{lemma}\label{lem:lds}
Suppose {\putaside}$(C)$ is run in all low-slack high-degree almost-cliques $C$, returning a set $P_C$. 
Then, for each such $C$, $|P_C| = \Omega(\ell^2)$, w.p.~$1-\exp(-\Omega(\ell))$.
\end{lemma}
\begin{proof}
By a standard Chernoff bound,  $|S_C|=\Omega( |C|/\ell)= c'\ell^2$, w.p.\ $1-\exp(-\Omega(\ell^2))$, for a constant $c'>0$, where we used the fact that $|C|=\Theta(\Delta_C)=\Omega(\ell^3)$.
    For a node $v\in S_C$, let $X_v$ be the indicator random variable that is 1 when an external neighbor of $v$ in a low-slack almost-clique $C'$ is sampled. Note that $P_C=\{v\in S_C : X_v=0\}$. 
    Since each node $w\in C'$ is sampled w.p.\ at most $\ell^2/(48\Delta_{C'})\le 1/(48\ell)$, and $v$ has external degree $e_v\le 12\barsigma_C\le 12\ell$ (by \cref{lem:degreesinac}), we have $\Pr[X_v=1]\le 12\ell/(48\ell)=1/4$. Note that each variable $X_v$ is a function of the independent indicator variables $\{Y_w\}_{w\in E_v}$,  of the events that an external neighbor $w$  is sampled. 
    Since every node $w\in C'$, with $C'\neq C$ and $\barsigma_{C'}\le \ell$, has at most $12\ell$ neighbors in $C$ (as external neighbors, as argued above),  
    we see that for a given $S_C$, 
    $\{X_v\}_{v\in S_C}$ is a read-$\ell$ family of random variables, and \cref{lem:kread} applies (with $q\leq 1/4$, $k=\ell$, $\delta=1/4$), showing that $|S_C|-|P_C|=\sum_{v\in S_C}X_v>  |S_C|/2$ holds w.p.\ less than $\exp(-\Omega(|S_C|/\ell))$. Thus, the probability that either $|S_C|<c'\ell^2$ or $|P_C|<|S_C|/2$ is $\exp(-\Omega(\ell))$. The proof follows.
\end{proof}

\subsection{Internal Degree Reduction}
\label{ss:sct}

Synchronizing color trials in dense components is fundamental to all known sublogarithmic-time $(\Delta+1)$-coloring algorithms. In \cite{HSS18}, such a primitive was applied $\sqrt{\log n}$ times; in \cite{HKMT21}, $\log^2 \log \Delta$ times; while 
in \cite{CLP20}, two such primitives were defined and applied $O(1)$ times in different ways on  different subgraphs. 
Here we apply only once a particularly na\"ive such primitive that avoids any communication about the topology or the node palettes.

In {\synchronizedcolortrial}, the leader sends a random unused candidate color from \emph{its own palette} to each inlier in $\core_C$ (to try).
Every node receives a distinct color, since by the definition of $O_C$, $\core_C\subseteq N(x_C)$.
It is \emph{a priori} far from obvious that such a simple primitive for coloring dense nodes has a hope of succeeding.

\begin{algorithm}[H]\caption{\synchronizedcolortrial, for almost-clique $C$}
\label{alg:synchtrial}
  \begin{algorithmic}[1]
    \STATE $x_C$ randomly permutes its palette $\pal(x_C)$, sends each neighbor $u \in \core_C$ a distinct color $\col_u$. \label{st:randomorder}
    \STATE Each $u \in \core_C$ calls {\trycolor}($u$, $\col_u$) if $\col_u \in \pal(u)$
    \end{algorithmic}
\end{algorithm}

We bound how many nodes in $C$ are \emph{decolored}, i.e.\ fail to get colored by {\synchronizedcolortrial} (Step~\ref{st:synchtrial} of \cref{alg:logstar-dense}).

The following lemma bounds the expected number of nodes that are decolored because they received from $x$ a color outside their palette.
Recall that $\disc(a,b)=|\Psi(a)\setminus \Psi(b)|/|\Psi(a)|$ denotes the current disparity of node $a$ toward $b$, which equals the probability that a random color from the palette $\Psi(a)$ falls outside $b$'s palette, $\Psi(b)$.

\begin{lemma} After slack generation, $\sum_{u \in \core} \disc(x,u) \le 44 \barsigma_C$ holds w.p.\ $1-\exp(-\Omega(d_x))$.
\label{L:inliers}
\end{lemma}
\begin{proof}
By a standard Chernoff bound, slack generation colors less than half of $x$'s neighbors, w.p.\ $1-\exp(-\Omega(d_x))$.  This allows us to bound the sum of the disparities of the leader $x$ towards nodes in $\core$ after slack generation in terms of original quantities:
\begin{align*} \sum_{u\in \core} \disc(x,u) 
  &= \sum_{u \in \core} \frac{|\Psi(x) \setminus \Psi(u)|}{|\Psi(x)|} 
  \le \sum_{u \in \core} \frac{|\Psi(x) \triangle \Psi(u)|}{|\Psi(x)|} 
  \le \sum_{u \in \core} \frac{|\Psi_x \triangle \Psi_u|+12\sigma_x}{|\Psi(x)|}\\
  & \le 2 \sum_{u \in \core} \frac{2|\Psi_u \setminus \Psi_x|+18\sigma_x}{(1-\eacd)|I|}
   \le \frac{4\disc_x}{1-2\eacd} + \frac{36 \sigma_x}{1-\eacd} \le 44 \barsigma_x
  \ , \end{align*}
where we used, respectively, the definition of disparity, the definition of symmetric difference, \cref{lem:degreesinac}, \cref{L:slackgen-dense} and \cref{lem:degreesinac}, the definition of discrepancy, and finally that $\eacd \le 1/13$.
\end{proof}

The following key lemma shows that a single {\synchronizedcolortrial} suffices to reduce the size of an almost-clique to its sparsity, paving the way for the application of {\slackcolor}.

\begin{lemma}\label{lem:degred}
Let $C$ be an almost-clique, and let $t=\Omega(\barsigma_C)$. W.p.\ $1-\exp(-t)$, the number of decolored nodes of $C$ in step~\ref{st:synchtrial} in \cref{alg:logstar-dense} is $O(t)$. 
\end{lemma}
\begin{proof}
Fix arbitrary candidate colors for nodes outside $C$ -- we prove the success of the algorithm within $C$ for arbitrary behaviors outside $C$.
Let $\core'$ be an arbitrary subset of $\core_C$ of size $\ceil{\card{\core_C}/2}$. 
    For each $u \in \core'$, recall that $\col_u \in \pal(x_C)$ is its candidate color, and let $X_u$ be the binary r.v.\ that is 1 iff $u$ is decolored. 
    Consider a node $u\in \core'$. Conditioning on an arbitrary set $S$ of $|\core'|-1$ candidate colors assigned to the nodes in $\core'\setminus u$, $\col_u$ is uniformly distributed in $\pal(x_C)\setminus S$, which has size $|\pal(x_C)|-|\core'|+1\ge  |\core_C|/2$ (since $\core_C\subseteq N(x_C)$). The node $u$ is decolored only when its candidate color is also tried by one of its external neighbors, of which it has $e_u = |E_{u}|=O(\barsigma_C)$ (\cref{lem:low-ext-degree}), or when it is not in its palette, i.e., when it belongs to $\pal({x_C})\setminus \pal(u)$.
    Thus, $\Pr\event{X_u=1\mid \{\col_w\}_{w\ne u}}\le q_u$, for $q_u = |\pal(x_C) \setminus \pal(u)|/|\pal(x_C)| + e_u/|\core_C| =  \disc(x_C,u)+O(\barsigma_C/|\core_C|)$. 
    Having fixed the candidate colors of nodes outside $C$, each $X_u$ is determined by $\col_u$, so we also have $\Pr\event{X_u=1\mid \{X_w\}_{w \ne u}}\le q_u$.
    Note that $\sum_{u \in \core'} q_u = \sum_{u \in \core'} \disc(x_C,u) + O(\barsigma_C) = O(\barsigma_C)$, by \cref{L:inliers}.
    Applying \cref{lem:chernoff}, we get that
    $\Pr\event{\sum_{u \in \core'} X_u\le 4t}\ge 1-\exp(-t)$ for any $t \ge \sum_{u \in \core'} q_u =  \Theta(\barsigma_C)$. By symmetry, the same holds for $\core_C \setminus \core'$, and the lemma follows by the union bound.
\end{proof}

\section{Combined Algorithm}

By first coloring the sparse nodes, then the dense nodes, using the results from the previous sections, we get an algorithm (\cref{alg:logstar-both}) that we can apply to a whole graph with degrees in a range $[\log^7 \Delta, \Delta]$. 

\begin{algorithm}[H]
\caption{Randomized $\deg+1$-Coloring Algorithm ($\forall v, d_v \in[\log^7 \Delta,\Delta]$)}
\label{alg:logstar-both}
  \begin{algorithmic}[1]
    \STATE \computeacd.
    \STATE Apply \cref{alg:logstar-sparse} to sparse nodes.
    \STATE Apply \cref{alg:logstar-dense} to dense nodes.
  \end{algorithmic}
\end{algorithm}

\maintechtheorem*

\begin{proof}
Observe that for color trial-based algorithms as ours, it can only help the algorithm to have a subset of nodes not participate in the algorithm; 
therefore, 
we may w.l.o.g.\ focus on the case when $d_v\ge \log^7\Delta$ holds for all $v\in V$, i.e., $V_H=V$. The ACD is computed deterministically. Given the lower bound on  degrees, \cref{L:slackgen-sparse} implies that each sparse node $v$ gets $\Omega(d_v)$ slack w.p.\ $1-\Delta^{-c}$, either via slack generation or by being  treated while $\Omega(d_v)$ of its uncolored neighbors stay out. The theorem then follows for sparse nodes by \cref{lem:slackcolor}. Let us turn to dense nodes now.

Recall that an almost-clique $C$ is \emph{high-slack} if $\barsigma_C = \Omega(\ell)$ and \emph{low-slack} otherwise.
Consider first a high-slack almost-clique $C$.
Each node in $C$ gets slack $\Omega(\barsigma_v)$ by \cref{L:slackgen-dense} (w.p.\ $1-\exp(-\Omega(\ell))$), which is $\Omega(\barsigma_C)$ by definition of $\barsigma_C$.
After \synchronizedcolortrial, $C$ contains $O(\barsigma_C)$ uncolored nodes, w.p.\ $1-\exp(-\Omega(\barsigma_C)) = 1 - \exp(-\Omega(\ell))$ (\cref{lem:degred}). Thus, $v$ has internal degree $O(\barsigma_C)$ and by \cref{lem:degreesinac}, it has also external degree $O(\barsigma_C)$. Thus, assuming the probabilistic statements hold, the slack of $v$ is at least proportional to its degree, and hence it gets colored by $\slackcolor(C\setminus P)$, w.p.\ $1-\exp(-\Omega(\ell))$.

Consider now a low-slack almost-clique $C$.
Each node $v$ in $C\setminus P$ gets temporary slack $|P|-a_v = \Omega(\ell^2) = \Omega(\ell)$ from $P\cap C$, by \cref{lem:degred}. After \synchronizedcolortrial, $C$ contains $O(\ell)$ nodes, w.p.\ $1-\exp(-\Omega(\ell))$. Thus, $v$ has internal degree $O(\ell)$, and it also has external degree $O(\barsigma_C) = O(\ell)$ (\cref{lem:degreesinac}). Thus, $C\setminus P$ gets colored by $\slackcolor(C \setminus P)$, w.p.\ $1-\exp(-\Omega(\ell))$.
All of the above probabilistic events hold even if random bits of nodes outside the 2-hop neighborhoods of $v$ are chosen adversarially.

The remaining nodes in $P$ are dense clusters with no interconnecting edges. They can be colored locally in two rounds. 
\end{proof}

\paragraph{Algorithm for general graphs.}

To turn the result above into an algorithm for coloring all graphs, we need on one hand to apply it to the whole range of possible degrees, and on the other hand, to deal with nodes that fail to get colored by \cref{alg:logstar-both}.

To color the whole graph, we treat the graph in degree ranges. We define $\Delta^{(i)}$ by $\Delta^{(0)} = n$ and $\Delta^{(i+1)} = \log^7 \Delta^{(i)}$ -- note that the nodes do not need to know the maximum degree of the graph. Initially, $\Delta^{(0)}=n$ is an upper-bound on the maximum degree of the (uncolored part of the) graph. We apply \cref{thm:main_tech} $\log^* n$ times, lowering this upper-bound on the maximum degree of the graph from $\Delta^{(i)}$ to $\Delta^{(i+1)}$ in iteration $i$. This is achieved by having all the nodes with a degree in the range $[\Delta^{(i+1)}, \Delta^{(i)}]$ run the algorithm and get colored w.h.p., paving the way to the next iteration as no node of degree at least $\Delta^{(i+1)}$ remains uncolored. Note that $\Delta^{(\log^* \Delta)} = O(1)$, so after the $\log^* n$ iterations, the remaining nodes are of constant degree.

The resulting algorithm is \cref{alg:whole}. In every iteration, each node that fails the probabilistic guarantee of \cref{thm:main_tech}, joins a set $\BAD$. To handle nodes that fail, we run a deterministic algorithm on $\BAD$ within each degree range. Each such instance consists of poly-logarithmic-sized components, as we argue in the following subsection, allowing for a fast solution.

\begin{algorithm}[H]
\caption{Randomized $\deg+1$-Coloring Algorithm (general graphs)}
\label{alg:whole}
  \begin{algorithmic}[1]
    \FOR{$i=1$ to $\log^* n$ in sequence}
        \STATE Apply \cref{alg:logstar-both} to nodes with current uncolored degree at least $\Delta^{(i)}$.
        \STATE Call a deterministic algorithm on $\BAD$.
    \ENDFOR
    \STATE The remaining instance has constant degree and can be colored in $O(\log^* n)$ rounds \cite{linial92}.
  \end{algorithmic}
\end{algorithm}

\paragraph{Shattering.}

Suppose we are in some iteration $i$ and let $\Delta = \Delta^{(i)}$.
Whenever a node $v$ fails a probabilistic guarantee, it is placed in $\BAD$. This can occur when failing to generate the promised slack in {\slackgeneration} or {\putaside}; {\slackcolor} failing to color all the respective nodes, or {\synchronizedcolortrial} leaving more than $O(\barsigma_v)$ nodes uncolored (\cref{lem:degred}).

When $\log \Delta \geq \cshat \cdot \log n$, for a large enough constant $\cshat$, $\BAD$ is empty w.h.p., due to \cref{thm:main_tech}. We stop the algorithm there in this case, otherwise we solve the \degoLC subproblem induced by $\BAD$.

\begin{lemma}
The probability that a node $v$ is added to $\BAD$ is $\Delta^{-\omega(1)}$, even if the random bits outside the $2$-hop neighborhood are determined adversarially.
\label{claim:bad}
\end{lemma}

\begin{proof}
\slackgeneration fails with probability at most $\exp(-\Omega(\ell)) = \Delta^{-\omega(1)}$ by \cref{L:slackgen-sparse,L:slackgen-dense},
and the same holds for \putaside (\cref{lem:lds}) and \synchronizedcolortrial (\cref{lem:degred}). These algorithms run in at most 2 rounds, so cannot depend on anything beyond the 2-hop neighborhood.

\slackcolor fails with probability $\exp(-\Omega(\smin^{1/(1+\kappa)}) )$ by \cref{lem:slackcolor}, where $\smin$ is the minimum slack, and $\kappa>1/\smin$. In our case, \slackcolor is applied only when the slack is $\Omega(\ell)$, resulting in failure probability $\exp(-\omega(\log \Delta)) = \Delta^{-\omega(1)}$, choosing a constant $\kappa < \delta$.
By \cref{lem:slackcolor}, the bound holds even if random bits outside the 2-hop neighborhood are determined adversarially.
\end{proof}

We use the following shattering lemma from \cite{CLP20}.

\begin{proposition}[Lemma 4.1 of \cite{CLP20}]
Consider a randomized procedure that generates a subset $\BAD \subseteq  V$ of vertices. Suppose that for each $v \in  V$, we have $\Pr[v \in  \BAD] \leq \Delta^{-3c}$, and this holds even if the random bits outside of the $c$-hop neighborhood of $v$
are determined adversarially.
W.p.\ $1-n^{-\Omega(c')}$, each connected component in $G[\BAD]$ has size at most $(c' /c)\Delta^{2c} \log_\Delta  n$.
\label{prop:shattering}
\end{proposition}

The next lemma follows from \cref{claim:bad} and \cref{prop:shattering}.

\begin{lemma}
$\BAD$ induces a subgraph whose connected components are of size $O(\ell^4 \cdot \log n)$. 
\label{L:bad}
\end{lemma}

A subtlety has to be addressed, which is that the graph induced by the subset of nodes of degree in the range $[\Delta^{(i+1)}, \Delta^{(i)}]$ can have much smaller degrees than $\Delta^{(i+1)}$, so we no longer are in a D1LC instance with the hypotheses asked by our algorithm. Nevertheless, the palettes of these nodes are still of size at least $\Delta^{(i+1)}$. Nodes whose degree decreased but stayed above $\Delta^{(i+1)}$ can simply throw away some colors to have palettes of size degree+1. Nodes whose degrees dropped below $\Delta^{(i+1)}$ in the induced subgraph can pretend to have more neighbors and 2-hop neighbors than they actually have so that the induced graph augmented with these virtual neighbors satisfies the degree requirements. These virtual neighbors are introduced such that their only connection to the real subgraph is the node that invented them. The algorithm can be run on this augmented graph by having each node that invented virtual neighbors simulate the algorithm for the nodes it invented. This causes no added communication and the impact on the size of the graph is minimal.

\paragraph{Analysis of round complexity.}
Let $T(n,d)$ be the optimal round complexity of deterministic distributed $\deg+1$-coloring algorithms on graphs with $n$ nodes and maximum degree $d$. Let $T(n) = \max_d T(n,d)$.
Currently, the best bounds known for $T(n,d)$ are $O(\log^2 d \log n)$ \cite{GK21} and $\tilde{O}(\sqrt{d}) + O(\log^* n)$ \cite{Barenboim16}. 

The complexity of our algorithm, outside the induced subproblems, is $O((\log^* \Delta)^2)$, or $O(\log^* \Delta)$ on each degree group.
The total round complexity of our algorithm is
\[\hat{T}(n) = \sum_{i = 1}^{\log^* n} T(\log n, \log^{(i)} n) \ . \]
This is $O(T(\log n))$ whenever $T(n) = \Omega(\log^{(k)} n)$, for some constant $k$, since $T(\log n,\log^{(i+1)} n) = O(\log^{(i+1)} n)$ for $i > k$. 
In particular, given the deterministic algorithm of \cite{GK21}, the complexity is $\hat{T}(n) = O(\log^3\log n)$, proving our main result:

\maintheorem*

\section{Slack Generation: Technical Details}
\label{S:slackgen}

In this section, we prove our statements about \slackgeneration.
\Crefrange{ss:bal}{ss:disc}
involve properties that are relevant to both sparse and dense nodes, while \cref{ss:tough} and \cref{ss:slackgen-dense} contain the full arguments for the sparse and dense nodes, respectively.

We define and use a number of small epsilon constants in the upcoming arguments. For reference, here are their order of magnitude in relation to $\eacd$: $\espa,\eub \in \Theta(\eacd^2)$; $\ehat \in \Theta(\eacd^4)$; $\ehc \in \Theta(\eacd^8)$.
Constraints are $\espa - \eub \in \Omega(1)$, $\eub > cte \cdot \sqrt{\ehat}$, $\ehc \leq \ehat^2/540$. $\eacd \le 1/125$. $\ehat \le 1/36$.

\subsection{Slack from Balanced Sources}
\label{ss:bal}

\begin{definition}[Balanced/Unbalanced]
Let $\eub \in (0,1)$. A node $v$ is
    \emph{$\eub$-balanced} if it has at least $(1-\eub) d_v$ neighbors $w \in N(v)$ with $d_w \ge 2 d_v/3$,
    and otherwise it is \emph{$\eub$-unbalanced}. 
    \label{def:top-heavy-balanced}
\end{definition}
We set two constants $\ehat, \eub \in (0,1)$ and consider
a node to be \emph{balanced} (\emph{unbalanced}, 
\emph{discrepant}) if it is $\eub$-balanced ($\eub$-unbalanced, 
$\ehat$-discrepant, respectively).

Let $\Nbal(v) = \set{u \in N(v): d_u \geq 2d_v/3}$ and $\dbal_v = \card{\Nbal(v)}$. Let us decompose each node's discrepancy and sparsity into the part that is contributed by $\Nbal(v)$ and the rest:  $\discbal_v = \sum_{u \in \Nbal(v)} \frac{\card{\pal_u \setminus \pal_v}}{\card{\pal_u}}$ and $\discunb_v = \disc_v - \discbal_v$; $\sparbal_v = \frac 1 {d_v} \parens*{\binom{\dbal_v}{2}-m(\Nbal(v))}$ and $\sparunb_v = \spar_v - \sparbal_v$.
We first show that discrepancy and sparsity coming from $\Nbal(v)$ easily give slack.

\begin{lemma}
\label{lem:balanced-disc}
A node $v$ of balanced discrepancy $\discbal_v$ receives $\Omega(\discbal_v)$ slack w.p.\ $1-\exp(-\Omega(\discbal_v))$ during \slackgeneration.
\end{lemma}
\begin{proof}
Let $\overline{\pal}=\bigcup_{u\in \Nbal(v)}\pal_u\setminus \pal_v$. For each color $c\in \overline{\pal}$, let us consider $A_c$, the event that some node in $\Nbal(v)$ tries $c$. Let $B_c$ be the event that $c$ is successfully tried by a node $u$ in $\Nbal(v)$, meaning that no node in $N^{+}(u)\cup \Nbal(v)$ tries the same color.
Note that $\discbal_v$ is the expected number of nodes in $\Nbal(v)$ that try a color in $\overline{\pal}$. Let $h=\sum_c \mathbb{I}_{A_c\cap B_c}$ be the number of successfully tried colors, which is a lower bound on the slack that $v$ gets during \slackgeneration.
Each color tried by a node $u$ is successful w.p.\ \[\left(1-\frac{1}{d(u)}\right)^{|N^+(u)|}\cdot \left(1-\frac{3}{2d(v)}\right)^{|\Nbal(v)|}=\Omega(1)\ ,
\]
which implies that the expected number of nodes successfully trying a color from $\overline{\pal}$ is $\Omega(\discbal_v)$, which in turn implies $\Exp[h]=\Omega(\discbal_v)$ (recall that successful trial means that the color is not tried by any other node in $\Nbal_v$).
Let $f=\sum_c \mathbb{I}_{A_c}$ and $g=f-h$. We have $\Exp[f]\ge \Exp[h]=\Omega(\discbal_v)$, and observe that $f$ and $g$ are $1$-Lipschitz and $2$-certifiable functions of the random color choices of all nodes, hence \cref{lem:talagrand-difference} applies to this setup, implying that $v$ gets $\Omega(\discbal_v)$ slack w.p.\ $1-\exp(-\Omega(\discbal_v))$.
\end{proof}

Note that \cref{lem:balanced-disc} immediately implies that uneven nodes get $\Omega(d_v)$ slack w.p.\ $1-\exp(\Omega(d_v))$ since $\discbal_v \geq \unev_v$.

\begin{lemma}
\label{lem:balanced-spar}
A node $v$ of balanced sparsity $\sparbal_v$ receives $\Omega(\sparbal_v)$ slack w.p.\ $1-\exp(-\Omega(\sparbal_v))$ during \slackgeneration.
\end{lemma}
\begin{proof}
If $\discbal_v > \sparbal_v/72$, \cref{lem:balanced-disc} already implies the result. Suppose now that $\discbal_v\leq \sparbal_v/72$. Since $\discbal_v = \sum_{u\in \Nbal(v)} \disc_{u,v}$, $v$ has at most $36\discbal_v$ neighbors $u\in\Nbal(v)$ for which $\disc_{u,v}=\card{\pal_u\setminus\pal_v}/ \card{\pal_u} \geq 1/36$. Let us ignore these nodes, and let $N'\subseteq \Nbal(v)$ be the remaining set of at least $\dbal_v-36\discbal_v \geq \dbal_v - \sparbal_v/2$ balanced neighbors of $v$. Such neighbors $u\in N'$ have degree between $2d_v/3$ (since they belong to $\Nbal(v)$) and $36d_v/35$ (because of the condition $|\pal_u\setminus\pal_v|<|\pal_u|/36$), and $\card{\pal_u\cap\pal_v} \geq 23\card{\pal_v}/36$ shared colors with $v$ (using $d_u\ge 2d_v/3$ and the same condition on palettes). This last result means that each pair of such neighbors shares at least $5\card{\pal_v}/18$ colors from $\pal_v$, therefore has a chance $\Omega(1/d_v)$ of trying the same color in \slackgeneration. Since $\Nbal(v)$ contains $\sparbal_v\dbal_v$ non-edges, and $\card{\Nbal(v)\setminus N'}\leq \sparbal_v / 2$, there are at least $\sparbal_v \dbal_v / 2$ non-edges between nodes in $N'$. Hence, the expected number of non-edges in that try the same color on both endpoints in $\Nbal(v)$ is $\Theta(\sparbal_v)$.

Let us consider such a non-edge to be successful if it gets to keep its color on both ends and no other non-edge in $\Nbal(v)$ tries the same color. Each successful non-edge contributes $1$ to the slack of $v$. A non-edge whose endpoints try the same color is successful w.p.\ $\Omega(1)$ (by an analogous argument as in \cref{lem:balanced-disc}), so the expected number of successful non-edges is $\Theta(\sparbal_v)$. For each color $c$ we define $A_{c}$, the event that the endpoints of some non-edge in $\Nbal(v)$ tried $c$, and $B_{c}$, the event that $c$ was  tried and successful. Let $f$ ($g$) be the sum of the indicator functions of the $A_{c}$'s ($A_c\cap \overline{B}_c$'s respectively), and $h=f-g$ be the number of colors that were successfully tried by a non-edge in $\Nbal(v)$. We have that $\Exp[f]\ge \Exp[h]=\Omega(\sparbal_v)$, by the argument above (see also \cref{lem:balanced-disc}). 
Moreover, $f$ and $g$ are both $1$-Lipschitz and $3$-certifiable, so by \cref{lem:talagrand-difference}, $v$ gets $\Omega(\Exp[h])=\Omega(\sparbal_v)$ slack w.p.\ $1-\exp(-\Omega(\sparbal_v))$.
\end{proof}

The two previous lemmas immediately imply the next one.

\begin{lemma}
If $v$ is sparse and balanced, then after slack generation, $v$ gets slack $\Omega(d_v)$, w.p.\ $1-\exp(-\Omega(d_v))$.
\label{L:balanced}
\end{lemma}
\begin{proof}
A sparse and balanced node has balanced sparsity $\sparbal_v \geq (\espa-\eub) d_v$, hence the result by \cref{lem:balanced-spar}.
\end{proof}

\subsection{Heavy Colors}
\label{sss:heavy}

The challenge with heavy colors is that each of them can contribute a large amount of slack in expectation, and a change in the color of a single node can decrease the expected total contribution of other nodes significantly. Thus, the usual concentration bounds do not apply. 
We tackle this by a two-stage analysis, grouping colors into buckets, and considering the contribution to slack of each bucket. Conditioned appropriately, those contributions are independent from each other, and we can argue concentration for their sum.

Here we consider the general case where we would like to argue slack of nodes w.r.t.\ a subset $S_v$ of their neighbors.
We apply it in the sparse case with $S_v = N(v)$ and in the dense case with $S_v = E_v$.
To that end, we restate the definitions of several key concepts related to slack generation in terms of $S_v$:
\begin{definition}[Color weight, heavy/light colors and nodes]
    Let $\ehc,\eps \in (0,1)$. A color $c \in \colspace$ is $\ehc$-\emph{heavy} for $v$ w.r.t. $S=S_v$ if its \emph{weight} $H(c)=H^{(S)}_v(c)=\sum_{u\in S, c\in \pal_u}\frac{1}{|\pal_u|}$ satisfies $H(c)\ge 1/\ehc$. Otherwise, it is $\ehc$-\emph{light}. We denote by $\hcol=\hcolS{(S)}_v$ ($\lcol=\lcolS{(S)}_v$) the sets of $\ehc$-heavy ($\ehc$-light) colors for $v$ in $S$. 
    Node $v$ is $\eps$-\emph{heavy} iff $\sum_{c \in \hcolS{(S)}_v} H^{(S)}(c)\ge \eps \card{S}$, and is otherwise $\eps$-\emph{light}.
    \label{def:heavy-colors}
\end{definition}

We set an additional constant $\ehc \in (0,1)$, $\ehc \leq \ehat^2/540$ and consider
a color to be \emph{heavy} (\emph{light}) if it is $\ehc$-heavy ($\ehc$-light).
Let $\hcol=\hcolS{(S_v)}_v$ be the set of heavy colors for $v$ in $S_v$, and $H(c)=H^{(S_v)}(c)$ the weight of a color $c$ for $v$ within $S_v$.

\begin{observation}\label{obs:colorweightsum}
$\sum_{c\in \colspace} H(c) = \card{S_v}$ and $|\hcol|\le \ehc \card{S_v}$.
\end{observation}
\begin{proof}
The first claim follows by a sum rearrangement (over colors vs.\ over $S_v$). The second follows from the first since each heavy color has weight $H(c)\ge 1/\ehc$.
\end{proof}

\begin{lemma}
Suppose $v$ is $\eps$-heavy in $S_v \subseteq N(v)$, with
$\ehc \le \eps^2/240$.
Let $C_{\eps,\ehc} = \frac {10^3(6 \cdot 240)^2}{9 \cdot \eps^{3}\ehc^{2}}$. If a value $s\ge C_{\eps,\ehc} \ln^2 \Delta$ exists such that $\min(\set{d_u, u\in S_v}) \geq s$ and $|S_v| \ge s$,
then after slack generation, $v$ gets (permanent) slack $\Omega(\card{S_v})$, w.p.\ $1- \exp(-\Omega(\sqrt{s}))$.
\label{lem:slack-heavy}
\end{lemma}
\begin{proof}
Let $\gamma =1/10$.
Notice that $C_{\eps,\ehc} = \frac {(6 \cdot 240)^2}{\eps^3\ehc^2\gamma^2(1-\gamma)}$.
Let $k = (2\eps/\ehc)\sqrt{s/C_{\eps,\ehc}} \geq (2\eps/\ehc)\ln \Delta$.
Let $W = \eps\gamma^2(1-\gamma) s/(6k)$,
and w.l.o.g., assume for simplicity of exposure that $W$ is an integer.
Observe that $s/W \geq  (6/(\eps\gamma^{2}))\cdot k \geq (12/(\ehc\gamma^{2}))\cdot\ln\Delta$, that $W = (6\cdot 120^2 /\eps^{4})\cdot k$,
and that both $k,W \in \Theta(\sqrt{s})$.

We will focus on the "heavy-color-rich" subset $\Vhe$ of $S_v$, i.e., those 
neighbors in $S_v$ with at least $\eps/2$-fraction of their colors heavy: $\Vhe = \{w \in S_v: |\pal_w\cap \hcol| \ge \eps |\pal_w|/2 \}$. 
We claim that $|\Vhe| \ge \eps \card{S_v}/2$.
Since $v$ is heavy, we have 
\[
\sum_{c\in \hcol} H(c)=\sum_{c\in \hcol}\sum_{w\in S_v, c\in \pal_w} \frac{1}{|\pal_w|}=\sum_{w\in S_v}\frac{|\hcol\cap \pal_w|}{|\pal_w|}\ge \eps \card{S_v}\ .
\]
Since each summand in the last sum is in $[0,1]$, at least $\eps/2$-fraction of nodes $w\in S_v$ must satisfy $|\pal_w\cap \hcol_v|/|\pal_w|\ge \eps/2$, establishing the claim.

We now  fix a partitioning of the colorspace $\colspace$ into buckets with the properties described in the following claim.

\begin{claim}\label{claim:buckets}
There exists a partitioning of colors $\colspace$ into $W$ buckets $B_1,\dots,B_{W}$ such that for  every node $u\in N^2(v)$ and every bucket $B_i$, it holds that 
\begin{compactenum}
    \item $|B_i \cap \pal_u| = (1 \pm \gamma)|\pal_u|/W$,
    \item for $u \in \Vhe$, $|B_i\cap \pal_u \cap \hcol_v|/|\pal_u\cap B_i| = (1 \pm \gamma) |\pal_u \cap \hcol_v|/|\pal_u|$. (The fraction of $u$'s color in bucket $i$ that are heavy is about the overall proportion.)
\end{compactenum}
\end{claim}
\begin{proof}
Consider a random partitioning of $\colspace$ into $W$ buckets  where each color is assigned a uniformly random bucket.
By Chernoff, for every node $u \in V$ and every bucket $B_i$,
the number of colors in $u$'s palette that fall in bucket $i$ is within $(1+\gamma)$-factor of the mean: 
\[ \Pr[ |B_i \cap \pal_u| = (1 \pm \gamma)|\pal_u|/W  ] \ge 1 - 2 \exp\parens*{-\gamma^2 \frac{|\pal_u|}{3W}} \geq 1-2\Delta^{-4} \ , \]
For every $u \in \Vhe$, the proportion of heavy colors in its palette that fall in each bucket $i$ is approximately the same as the overall proportion: 
\[
\Pr\left[\frac{|\pal_u \cap \hcol_v \cap B_i|}{|\pal_u\cap B_i|} = (1 \pm \gamma) \frac{|\pal_u \cap \hcol_v|}{|\pal_u|} \right] \ge 1-2\exp\parens*{-\gamma^2 \ehc \frac{|\pal_u|}{6W}} 
\ge 1-2\Delta^{-2}\ .
\] 
using that $|\pal_u|/W \ge s/W \ge 12 \ln\Delta/(\gamma^2 \ehc)$.
By a union bound over nodes in $N^2(v)$ and the $W$ buckets, the probability that the two properties above hold for all nodes $u\in N^2(v)$ is at least $1 - \Delta^2 W \cdot 2\Delta^{-4} + \Delta W \cdot 2\Delta^{-2} >0$. Thus, a partitioning of the colorspace with claimed properties exists. 
\end{proof}

We next show that each node has about the expected number of neighbors within each bucket,
using standard Chernoff bounds.
For a node $u$, let $B(u)\in \{B_i\}_{i=1}^W$ denote the bucket where the color tried by $u$ belongs. 
Let $D_u^i = |\{w\in N(u): B(w) = B_i\}|$ be the number of neighbors of a node $u\in S_v$ whose color is in bucket $B_i$. 
Observe that for any node $u$, index $i\in [W]$ and $w\in N(u)$, $B(w)=B_i$ holds w.p.\ $|\pal_w\cap B_i|/|\pal_w|=(1\pm\gamma)/W$, by \cref{claim:buckets}; hence, we have $\Exp[D_u^i] \le (1 + \gamma)d_u/W$.
Thus, by Chernoff, for each fixed $u$ and $i$,
\begin{equation}
\Pr[D_u^i \le (1 + \gamma)^2 d_u/W] \ge 1-2\exp(-\gamma^2(1-\gamma)d_u/W) \ge 1-\exp(-3k)\ ,
\label{eq:Dui}
\end{equation}
and this holds for all
$u\in S$ and $i\in [W]$ simultaneously, w.p.\ $1-\Delta W\exp(-3k) > 1-\exp(-k)$.
Also, let $\Vhe_i = \{u\in \Vhe: B(u) = B_i\}$ be the nodes $u\in \Vhe$ whose color is in bucket $B_i$.
With an identical argument we have that for all $i$,
$\card{\Vhe_i} \le (1 +\gamma)^2\card{\Vhe}/W$, w.p.\ $1-\exp(-k)$.

What follows is conditioned on the high-probability events 
that each node in $S_v$ has about the same number of neighbors in each bucket, and that there are about the same number of $\Vhe$ nodes in each bucket.
We assume that for every node $u$, $B(u)$ is given, and its color $\col_u$ is random in $B(u)$. For a node $u$, let $N^+_B(u)$ denote the set of neighbors $w$ such that $d_w\ge d_u$ and $B(w)=B(u)$.

We now bound the expected number of nodes (in $\Vhe$) that get colored with heavy colors.
For a node $u\in \Vhe$, let $Z_u$ be the indicator random variable of the event that
$u$ picks and keeps a heavy color, i.e., that $u$ picks $\psi_u\in \hcol$ and all nodes in $N^+_B(u)$ pick a color different from $\psi_u$. 
Let $Z = \sum_{u \in \Vhe} Z_u$. 
Nodes choose their colors independently.
The probability that a given node $w\in N^+_B(u)$ picks a color other than $\psi_u$ 
is at least $p = 1-W/((1-\gamma)|\pal_u|)$, since, by \cref{claim:buckets} (and since $d_w\ge d_u$), it has at least $(1-\gamma)|\pal_w|/W \ge (1-\gamma)|\pal_u|/W$ colors in bucket $B(u)$. By the conditioning above, we also have $|N^+_B(u)|\le D_u^{B(u)}\le  (1+\gamma)^2 d_u/W$. 
Since $u\in \Vhe$, the probability that it picks a heavy color, i.e., $\psi_u\in \hcol_v$, is at least $\eps/2$.
Putting together,
\[ \Exp_\sigma[Z_u] \ge p^{|N^+_B(u)|}
\ge \frac{\eps}{2} \parens*{1-\frac{W}{(1-\gamma)(d_u+1)}}^{(1+\gamma)^2d_u/W} \\
 \ge \frac{\eps}{2} \cdot \exp \parens*{-\frac{2(1+\gamma)^2}{1-\gamma}} \ge \frac{\eps}{30} \ .  \]
Thus, 
 \[ \Exp[Z] = \sum_{u \in \Vhe} Z_u \ge |\Vhe|\cdot \eps/30 \ge \eps^2 \card{S_v}/60 \ . \]

We now derive concentration bounds on $Z$, the number of nodes colored with heavy colors.
The variables $Z_u$ of nodes in the same bucket are highly dependent, so we instead consider $Z_i$, the number of heavy-colored nodes of $\Vhe$ in bucket $B_i$. Importantly, the $Z_i$ are independent and $Z=\sum_{i\in [W]} Z_i$.
From the conditioning above, each $Z_i$ is bounded by $0 \le Z_i \le |\Vhe_i| \le (1+\gamma)^2|\Vhe|/W \le 2 \card{S_v}/W]$.
Applying 
Hoeffding's inequality (\cref{lem:chernoff-various-ranges}) with $a_i = 0$, $b_i = 2\card{S_v}/W$, $t = \eps^2 \card{S_v}/120$, and $\sum_i (b_i - a_i)^2
= 4\card{S_v}^2/W$, we obtain that 
\[ \Pr\event*{Z \ge \eps^2 \card{S_v}/120}
\ge 1 - 2\exp\parens*{\frac{\eps^{4}}{2\cdot 120^2} W} = 1 - 2\exp(-3k)\ , \]  
using the definition of $W$.

The number of nodes colored heavy minus the number of heavy colors used gives us slack.
The number of heavy colors is at most $\ehc \card{S_v}$, by \cref{obs:colorweightsum}. 
Thus, as long as $\eps^2/120 \ge 2 \ehc$, the slack obtained is at least $\card{S_v} \eps^2/240  = \Omega(\card{S_v})$, w.p.\ $1 - \exp(-\Omega(k))$.
We had conditioned on the bounds on the $D_u^i$ and $\Vhe_i$ holding. Taking this into account, the probability that $v$ fails to attains slack of at least $\card{S_v} \eps^2/120$ is at most $1-\exp(-\Omega(k)) = 1-\exp(-\Omega(\sqrt{s}))$.
\end{proof}

The light colors have the useful property that other neighbors are not too likely to pick them and with that destroy a successful edge.
\begin{lemma}
Let $S_v \subseteq N(v)$. Given that a pair $u,w\in S_v$ of nodes (a single node $u\in S_v)$) picked a color $c$ of weight $H^{(S_v)}(c)$ in $S_v$, the probability that no other node $w'\in N^+(u)\cup N^+(w)\cup S_v$ ($w'\in N^+(u)\cup N(v)$) picks $c$ is $\exp(-O(H^{(S_v)}(c)))$. If $c$ is light in $S_v$, this probability is $\Omega(1)$.
\label{lem:light}
\end{lemma}
\begin{proof}
The probability that no node in $N^+(u)$ picks $c$ is at least $\prod_{u'\in N^+(u)}(1-1/|\pal_{u'}|)\ge (1-1/|\pal_u|)^{d_u}\ge 1/e$. Same holds for $N^+(w)$ for the pair of nodes case.
For a color $c$, let $S(c)=\{u\in S_v : \pal_u\ni c\}$
The probability that no node in $S_v$ picks $c$ is at least $\prod_{u'\in S(c)}(1-1/|\pal_{u'}|)\ge 4^{-\sum_{u'\in S(c)}1/|\pal_{u'}|} = 4^{-H^{(S_v)}(c)}$. If $c$ is light in $S_v$, $H^{(S_v)}(c) \leq 1/\ehc$ so this probability is $\Omega(1)$.
\end{proof}

\subsection{Discrepancy}
\label{ss:disc}

We show that nodes get slack proportional to their discrepancy. Similar to above, we state the result in terms of a subset $S_v\subseteq N(v)$ of neighbors, also modifying the definition of discrepancy.
The discrepancy of node $v$ w.r.t.\ a set $S_v$ is
$\disc_v(S_v)=\sum_{u \in S_v} \disc_{u,v}$.
A node $v$ is $\eps$-\emph{discrepant} w.r.t.\ $S_v$ if $\disc_v(S_v) \geq \eps |S_v|$.

\begin{lemma}
Let $\eps \in (0,1)$ be a constant and $S_v \subseteq N(v)$. There is a constant $C_{\eps}$ such that, if a value $s\ge C_{\eps} \ln^2 \Delta$ exists such that $\min(\set{d_u, u\in S_v}) \geq s$ and $|S_v| \ge s$,
then a $\eps$-discrepant node w.r.t.\ $S_v$ gets slack $\Omega(|S_v|)$, w.p.\ $1-\exp(-\Omega(\sqrt{s}))$.
\label{L:disc-slack}
\end{lemma}
\begin{proof}
We show that $v$ gets slack both when its discrepancy is due to heavy colors, and when its discrepancy is due to light colors, where the weight of colors is defined w.r.t.\ $S_v$ and a color is heavy if $H(c) \geq 540/\eps^2$.

If $v$ is $2\eps/3$-heavy within $S_v$, then it gets $\Omega(|S_v|)$ slack with the announced probability by \cref{lem:slack-heavy}, applied with $\eps \leftarrow 2\eps/3$ and $\ehc \leftarrow (2\eps/3)^2/240$. Otherwise, a constant fraction of the discrepancy comes from colors light within $S_v$, since:
\[ \sum_{u\in S_v} \frac{\card{\lcol \cap \pal_u \setminus \pal_v}}{\card{\pal_u}}
    =  \sum_{u\in S_v} \frac{\card{\pal_u \setminus \pal_v}}{\card{\pal_u}} - 
       \sum_{u\in S_v} \frac{\card{\hcol \cap \pal_u \setminus \pal_v}}{\card{\pal_u}} 
    \geq \disc_v(S_v) - \frac{2\eps}{3} |S_v| \geq \frac{\eps}{3}|S_v|.\]

Let $\col_u$ be the color tried by $u$ in \slackgeneration. The bound above expresses that $\Exp[\card{\set{u \in S_v: \col_u \in \lcol\setminus\pal_v}}] \geq \eps |S_v|/3$.
For a light color outside of $v$'s palette $c \in \lcol \setminus \pal_v$,  we say it is \emph{successfully} tried if a node $u$ in $S_v$ tries it while no other node in $S_v\cup N^+(u)$ does. Let $h$ be the number of nodes in $S_v$ successfully trying a color in $\lcol \setminus \pal_v$. By \cref{lem:light}, any trial of a light color succeeds w.p.\ $\Omega(1)$, so $\Exp[h]\in\Omega(|S_v|)$. Note that $h$ equivalently counts the number of colors in $\lcol \setminus \pal_v$ that are successfully tried by some (necessarily unique) node in $S_v$.

Let $A_c$ be the event that some neighbor $u\in S_v$ tries it, and let $f = \sum_{c \in \lcol \setminus \pal_v} \bbI_{A_c}$.
Let $B_c$ be the event that $c$ is \emph{successfully} tried. By our previous definition of $h$,  $h = \sum_{c \in \lcol \setminus \pal_v} \bbI_{A_c \cap B_c}$. We showed above that $\Exp[h] \in \Omega(|S_v|)$. $h$ is a lower bound on the slack received by $v$ during \slackgeneration, since it counts a subset of the neighbors of $v$ that color themselves with a color outside $\pal_v$.

 Consider the number $g = f - h$ of light colors that are unsuccessfully tried in $S_v$. $f$ and $g$ are $1$-Lipschitz and $2$-certifiable, as we only need to reveal the random color choices of $2$ nodes to show that some light color outside of $v$'s palette was unsuccessfully tried by one of its neighbors.
 $\Exp[f] \geq \Exp[h]$ and $\Exp[f] \leq \eps|S_v|/3$ so $\Exp[f] \in \Theta(|S_v|)$.  By \cref{lem:talagrand-difference}, $\card{h - \Exp[h]} \in \Omega(\Exp[h])$ w.p.\ $1-\exp(-\Omega(\Exp[h]))$.  Hence in that case, $v$ gets $\Omega(|S_v|)$ slack w.p.\ $1-\exp(-\Omega(|S_v|)) \geq 1-\exp(-\Omega(\sqrt{s}))$.
\end{proof}

\subsection{Sparse and Uneven Nodes}
\label{ss:tough}
We prove here our claim about slack generated for sparse and uneven nodes.
Recall that $\Veasy$ consists of the sparse nodes that are either balanced or discrepant, along with the uneven nodes, and nodes with at least $\ehat$-fraction of their neighbors being dense.
$\Vst$ are the sparse nodes that are neither heavy nor in $\Veasy$ but have at least $\ehat$-fraction of their neighbors in $\Veasy$.

\slacklemmasparse*

\begin{proof}
Each sparse or uneven node that is adjacent to $\ehat^2 d_v/10^4$ dense nodes gets slack from them. For simplicity of exposition we ignore dense neighbors in the coming arguments, which is w.l.o.g.\ since they immediately give slack and considering $q\leq \ehat^2 d_v/10^4$ neighbors can only modify our quantities (discrepancy, sparsity, balance, heaviness) by $q$. 
An uneven node $v$ receives slack $\Omega(d_v)$,  
w.p.\ $1 - \exp(-\Omega(\sqrt{s}))$, by \cref{L:disc-slack}, and the same holds for a sparse discrepant node.
A balanced sparse node $v$ obtains slack $\Omega(d_v)$, w.p.\ $1-\exp(-\Omega(d_v))$, by \cref{L:balanced}.
A heavy sparse node gets slack $\Omega(d_v)$,  
w.p.\ $1 - \exp(-\Omega(\sqrt{s}))$, by \cref{lem:slack-heavy}.
Any node $v$ that is adjacent to $\ehat d_v$ nodes that are either balanced, discrepant, or uneven, is added to $\Vst$. 
The remaining \emph{tough} nodes 
are light, unbalanced, non-discrepant, and not in $\Vst$:
$\Vtough = \Vsp \setminus (\Veasy \cup \Vheavy \cup \Vst)$, 
As we show in the remainder of this subsection, culminating in \cref{lem:tough-slack}, tough nodes get slack $\Omega(d_v)$ w.p.\ $1-\exp(-\Omega(d_v))$.
\end{proof}

\paragraph{Tough nodes}
Intuitively, the 'tough nodes' are sparse nodes for which none of the other arguments work, which gives them particular properties. 
We show that these properties imply that many non-edges in their neighborhoods, in expectation, have their endpoints try the same (light) color, which then leads to slack by successful non-edges.

We bound this number of non-edges by arguing that a set of unbalanced and non-discrepant nodes needs to have edges with the outside of their set. The argument goes as follows: an unbalanced node needs to be adjacent to nodes of lower degree. A non-discrepant node is also non-uneven, it is therefore connected to few nodes of higher degree. Each unbalanced edge connects a node of lower degree with one of higher degree. The nodes in the neighborhood of a tough node simultaneously need to be adjacent to many unbalanced edges while mostly only being able to act as the endpoints of higher degree. This implies that they need edges to outside the neighborhood of $v$. This creates a number of non-edges in a tough node's neighborhood that, due to the tough nodes' low discrepancy and low heavy colors, each contribute a constant amount of slack in expectation.

 Our goal is to bound the expected number $Z_v$ of non-adjacent node pairs in $N(v)$ that pick equal light colors. Then, the generated slack can be lower-bounded along the lines of a standard argument.
 
 Consider the sums
 $\UP=\sum_{u\neq w\in N(u)}\frac{|\pal_u\cap \pal_w \cap L|}{|\pal_u|\cdot|\pal_w|}$ and $\UE=\sum_{u,w\in N(u), uv\in E}\frac{|\pal_u\cap \pal_w\cap L|}{|\pal_u|\cdot|\pal_w|}$ that measure the expected number of (unordered) pairs of nodes in $N(v)$ that pick equal light colors (i.e., from $L$), and the expected number of \emph{adjacent} pairs of nodes in $N(v)$ that pick equal light colors. Note that $Z_v=\UP/2-\UE$. We bound $Z_v$ by obtaining bounds on $\UP$ and $\UE$.

 \begin{lemma}
 If $v$ is tough, then the expected number $Z_v$ of non-adjacent node pairs in $N(v)$ that pick equal light colors is at least 
   $Z_v\ge [(1-\ehat)\eub/2 - 5 \ehat -2 \sqrt{\ehat}]d_v$.
 \label{lem:tough-expectancy}
 \end{lemma}
 \begin{proof}
 
 Let $\mathcal{C}_v=\cup_{u\in N(v)}\pal_u$. For a color $c$, let $N(c)=\{u\in N(v) : \pal_u\ni c\}$.  Recall that $H(c)=\sum_{u\in N(c)}1/|\pal_u|$. Let $L$ be the set of light colors. We have:
\begin{claim}
 $\sum_{c\in L} H(c)\ge (1-\ehat)d_v$, and $\sum_{c\in \pal_v} H(c) \ge (1-\ehat) d_v$.
\end{claim}
\begin{proof}
 The first claim follows from the definition of $L$ and \cref{obs:colorweightsum}, and the second from the definition of discrepancy, since $\sum_{c\in \mathcal{C}_v\setminus \pal_v} H(c) = \sum_{u\in N(v)}\frac{|\pal_u\setminus \pal_v|}{|\pal_u|}\le \ehat d_v$.
\end{proof}

\begin{claim}
$\UP\ge (1-4\ehat)d_v-\sum_{u\in N(v)}\frac{1}{|\pal_u|}$.
\end{claim}
\begin{proof}
We bound the augmented sum $\bUP = \UP + \sum_{u\in N(v)}\frac{1}{|\pal_u|}$, which adds a correction term for neighbors of $v$ with small palettes.
By a sum rearrangement, an application of Cauchy-Schwartz inequality and the observation above, we have: 
\begin{align*}
\bUP &=\sum_{u,w\in N(v)}\sum_{c\in \pal_u\cap \pal_w\cap L}\frac{1}{|\pal_u||\pal_w|}=\sum_{c\in L}\sum_{u,w \in N(c)} \frac{1}{|\pal_u||\pal_w|} = \sum_{c\in L} H(c)^2\\ 
&\ge \sum_{c\in L\cap \pal_v} H(c)^2\ge \frac{\left(\sum_{c\in L\cap \pal_v} H(c)\right)^2}{d_v^2}\ge (1-2\ehat)^2 d_v\ ,
\end{align*}
which implies the claim.
\end{proof}

\begin{claim}
$2\UE \le [1-((1-\ehat)\eub/2 - \ehat - 2\sqrt{\ehat})]d_v - \sum_{u\in N(v)}\frac{1}{|\pal_u|}$.
\end{claim}
\begin{proof}
Note that
\begin{align}
2\UE &\le \sum_{u\in N(v)}\sum_{w\in N(u)\cap N(v)} \frac{|\pal_u\cap \pal_w|}{|\pal_u||\pal_w|}\le \sum_{u\in N(v)}\sum_{w\in N(u)\cap N(v)} \frac{1}{|\pal_u|}
= \sum_{u\in N(v)} \frac{|N(u)\cap N(v)|}{|\pal_u|} \nonumber \\
& = \sum_{u\in N(v)}\frac{|N(u)|}{|\pal_u|} -\sum_{u\in N(v)}\frac{|N(u)\setminus N(v)|}{|\pal_u|}=\sum_{u\in N(v)}\frac{|\pal_u|-1}{|\pal_u|}-\sum_{u\in N(v)}\frac{|N(u)\setminus N(v)|}{|\pal_u|}\ . 
\label{eq:ue}
\end{align}
We need to lower-bound $\sum_{u\in N(v)}|N(u)\setminus N(v)|/|\pal_u|$. We do this in two steps. First, we show that for each node $u\in N(v)$, there is a subset $Q_u\subseteq N(u)$ such that $\sum_{u\in N(v)}|Q_u|/|\pal_u|=\Omega(d_v)$, then we show that $\sum_{u\in N(v)}|Q_u\cap N(v)|/|\pal_u|=O(d_v)$, with appropriate constants, so that their difference gives the lower bound.

We let $Q_u\subseteq N(u)$ be the set of neighbors $w\in N(u)$ with $d_w\le 2d_u/3$. By assumption, there is a set $P_v\subseteq N(v)$ of at least
$(1-\ehat)d_v$
nodes that are $\eub$-unbalanced, that is, $|Q_u|\ge \eub d_u\ge (\eub/2) |\pal_u|$.
Then, we have 
$\sum_{u\in P_v}|Q_u|/|\pal_u|\ge (1-\ehat)\eub d_v/2$. 

Next, because $v$ has few discrepant neighbors and higher degree neighbors give discrepancy,
there is a set $R_v \subseteq N(v)$ of size at least 
$(1-\ehat) d_v$
such that
for all $w \in R_v$, the set $T_w \subseteq N(w)$ of neighbors $u$ of $w$ with $d_w\le (1-2\sqrt{\ehat})d_u$ has size $|T_w|\le 2\sqrt{\ehat} d_w$. Since $2\sqrt{\ehat} \le 1/3$, for nodes $u,w\in N(v)$, $w\in Q_u$ ($d_w\le 2d_u/3$) implies $u\in T_w$ ($d_w\le (1-2\sqrt{\ehat})d_u$). Hence, for every $w\in R_v$, $\sum_{u: w\in Q_u} 1/|\pal_u|\le \sum_{u \in T_w} 1/|\pal_u|<2\sqrt{\ehat}$, and  $\sum_{u\in N(v)} |Q_u\cap N(v)|/|\pal_u|\le (d_v - \card{R_v}) +  \sum_{w\in R_v}\sum_{u \in T_w} 1/|\pal_u|\le \ehat + 2\sqrt{\ehat} d_v$, using sum rearrangement.

Putting together, we get $\sum_{u\in N(v)}|N(u)\setminus N(v)|/|\pal_u|\ge [(1-2\sqrt{\ehat})\eub/2 - \ehat - 2\sqrt{\ehat}]d_v$, which in light of (\ref{eq:ue}) implies the claim.
\end{proof}
The proof of the lemma now follows from the last two claims.
\end{proof}

As in~\cite{EPS15}, let us call \emph{successful non-edges} the non-edges in $N(v)$ whose endpoints picked the same color during \slackgeneration such that: no neighbors of the endpoints picked this color, and no other nodes in $N(v)$ picked the same color.

\begin{lemma}
A tough node $v$ gets slack $\Omega(d_v)$ w.p.\ $1-\exp(-\Omega(d_v))$ during \slackgeneration.
\label{lem:tough-slack}
\end{lemma}
\begin{proof}
In expectation, a tough node $v$ has $\Omega(d_v)$ non-edges trying the same light color in its neighborhood by \cref{lem:tough-expectancy}. By \cref{lem:light}, the expected number of successful non-edges in $N(v)$ is also $\Omega(d_v)$. From there the proof is a classical result, the same as that of \cref{lem:balanced-spar}: for each light color, introduce event $A_c$ and $B_c$ indicating (respectively) whether $c$ was tried by a non-edge in the neighborhood of the tough node, and whether said try was successful; introduce $f = \sum_{c \in \lcol} \bbI_{A_c}$, $h = \sum_{c \in \lcol} \bbI_{A_c \cap B_c}$, and $g=f-h$; argue that $f$ and $g$ are $1$-Lipschitz and $3$-certifiable together with the fact that $\Exp[f],\Exp[h]\in \Theta(d_v)$ to apply \cref{lem:talagrand-difference} and get the result.
\end{proof}

\subsection{Dense Nodes}
\label{ss:slackgen-dense}

We now derive the claim about slack for dense nodes:

\slacklemmadense*

\begin{proof}
To prove this claim, let us decompose the slackability of a dense node as such: $\barsigma_v = \spar_v + \disc_v = \sparbal_v + \sparunb_v + \discbal_v + \discunb_v$. Using previous results, we show that $v$ gets slack $\Omega(\barsigma_v)$ when one of $\sparbal_v, \discbal_v, \discunb_v$ is of order $\Omega(\barsigma_v)$. We then analyze the remaining case, which we call the \emph{gritty} nodes, which like the tough nodes in the sparse case occurs when slackability is mostly due to unbalanced sparsity.

As immediate property, because $C_v \subseteq \Nbal(v)$, $\Nunb(v) \subseteq E_v$ and we have $\sparunb_v \leq e_v$ and $\discunb_v \leq e_v$. Let $\eps \in (0,1/6)$ be a constant in what follows.

\Cref{lem:balanced-disc,lem:balanced-spar} imply that a dense node gets $\Omega(\barsigma_v)$ slack w.p.\ $1-\exp(\Omega(\barsigma_v))$ if $\discbal_v \geq \eps \barsigma_v$ or $\sparbal_v \geq \eps \barsigma_v$. When this is not the case, $\discunb_v + \sparunb_v \geq (1-2\eps)\barsigma_v$, so $\barsigma_v \leq 3 e_v$. We consider this to hold in what follows.

If $\discunb_v \geq \eps \barsigma_v \geq \eps e_v / 3$, then $v$ is $\eps/3$-discrepant within $E_v$, of size $e_v \geq \barsigma_v/3 \geq \log^3\Delta$. By \cref{L:disc-slack}, $v$ gets $\Omega(e_v)=\Omega(\barsigma_v)$ slack w.p.\ $1-\exp(-\Omega(\log^{4/3}\Delta))$.

This proves the proposition when one of the quantities $\sparbal_v, \discbal_v, \discunb_v$ is of order $\Omega(\barsigma_v)$. By \cref{lem:slack-heavy}, $v$ also gets slack if those quantities are small (implying $\barsigma_v \in \Theta(e_v)$) and it has heavy colors within $E_v$. The remaining case -- when none of these arguments applies -- is covered by \cref{lem:gritty}, whose proof constitutes the rest of this section.
\end{proof}

\paragraph{Gritty nodes.}
The last case to consider is when the slackability is due to the sparsity created by the external neighbors.

Consider two constants $\ehc,\ehat \in (0, 1/6)$ with $\ehc \leq \ehat^2/288$. 
A node $v$ is \emph{gritty} if $\discbal_v \leq \ehc\barsigma_v/720$, $\sparbal_v\leq\ehat\barsigma_v$ and has $\ehat$-few $\ehc$-heavy colors within $E_v$, i.e., $\sum_{c\in \hcol} H^{(E_v)}(c)\le \ehat e_v$.
Note that the node would get $\Omega(\barsigma_v)$ if one of those conditions was not satisfied by previous arguments. We have $\sparunb_v\in \Omega(\barsigma_v)$, and since $e_v \geq \sparunb_v$, $e_v \in \Omega(\barsigma_v)$.

Intuitively, each lower-degree neighbor in $E_v$ has only a few neighbors in $C_v$ and has therefore are plenty of incident non-edges in $N(v)$. The issue is to ensure that the intersection of the palettes of the nodes of such a pair be large enough.

Since $v$ is not an outlier, there exists a set $U$ of $\card{C_v}/6$ nodes in $C_v$ with degree at least $d_v$. By the definition of almost-cliques (\cref{def:acd}), $v$ is adjacent to at least $(1/6 - \eacd)\card{C_v}$ of them. In the argument that follows, we identify a subset $C'_v \subseteq U$ that share a large part of $v$'s palette, and then show that an expected $\Omega(e_v)$ non-edges between $E_v$ and $C'_v$ try the same light color. This results in \cref{lem:gritty}.

\begin{lemma}
\label{lem:gritty}
Assume $\eacd \le 1/125$. A gritty node $v$ 
gets slack $\Omega(\barsigma_v)$ in slack generation, w.p.\ $1-\exp(-\Omega(\barsigma_v))$.
\end{lemma}
\begin{proof}
The conditions on $\sparbal_v$ and $\discbal_v$ imply that $\barsigma_v \leq 3e_v$, so we show the equivalent result that $v$ gets $\Omega(e_v)$ slack with $1-\exp(-\Omega(e_v))$.

Let $C'_v = \{w \in C_v \cap N(v) : d_w \ge d_v \text{~and~} |\pal_w\cap \pal_v| \ge d_v - \ehc e_v/20\}$ be the set of neighbors of $v$ in $C_v$ of at least as high degree and sharing at least $d_v - \ehc e_v/20$ colors with $v$.
\begin{claim}
It holds that $|C'_v| \ge |C_v|/15$.
\label{claim:cprime-size}
\end{claim}
\begin{proof}
Note that $C'_v \subseteq \Nbal(v)$, so the bound on $\discbal_v$ gives us information about the palettes of nodes in $C'_v$.
Let $U'\subseteq U$ denote the set of nodes in $U$ that have fewer than  $d_v -\ehc e_v/20$ colors in common with $v$; note that $C'_v=U\setminus U'$.
Each node in $U'$
has at least $d_u+1-d_v+\ehc e_v/20$ colors outside $v$'s palette, hence contributes $(d_u+1-d_v+\ehc e_v/20)/(d_u+1)\ge \ehc e_v/(20d_v)$ (observe that for numbers $a<b$ and $c>0$, it holds that $(a+c)/(b+c)>a/b$) to $\discbal_v$. The total contribution of $U'$ to $\discbal_v$ is therefore at least $\ehc|U'|e_v/(20 d_v)$. Since $\discbal_v \leq \ehc \barsigma_v/720 \leq \ehc e_v/240 $, we have $\ehc|U'|e_v/(20d_v)\le \ehc e_v/240$, hence $|U'|\le d_v/12\le (1+\eacd)|C_v|/12$, and $|C'_v| \ge (1/6-\eacd - (1+\eacd)/12)|C_v| \ge |C_v|/15$, using $\eacd<1/75$.
\end{proof}

We say that a neighbor $w \in E_v$ of $v$ is \emph{good} if $\sum_{u \in C'_v \cap \overline{N}(w)} |\pal_u \cap \pal_w \setminus \hcol_v| \ge |\pal_w|\cdot |C_v|/20$.
Intuitively, a good node has large potential for creating a same-colored pair with a non-neighbor in $C'_v$ using a light color.
In fact, the probability that this happens is $\Omega(1)$.

\begin{claim}
Node $v$ has $\Omega(e_v)$ good neighbors.
\label{claim:goodneighborsdense}
\end{claim}

\begin{proof}
First, we eliminate nodes in $E_v$ with too many heavy colors and then those with too many colors that occur infrequently in $C'_v$. The rest have mostly frequent light colors and are shown to be good.

Let $H_v(c) = H^{(E_v)}_v(c)$, $\hcol_v = \hcolS{(E_v)}_v$, $\lcol_v = \lcolS{(E_v)}_v$. Since $v$ does not have heavy colors, 
\[
\sum_{w\in E_v}\frac{|\hcol_v\cap \pal_w|}{|\pal_w|}
= \sum_{c \in \hcol_v} \sum_{w\in E_v: c\in \pal_u} \frac{1}{|\pal_w|} 
= \sum_{c\in \hcol_v} H_v(c)
\le \ehat e_v,
\]
and hence at most $4\ehat e_v$ nodes $w\in E_v$ have at least a quarter of their colors among $v$'s heavy colors: $|\pal_w\cap \hcol_v|/|\pal_w|\ge 1/4$.
We eliminate such nodes to obtain $E'_v$, where $|E'_v| \ge (1-4\ehat)e_v$.

Next, consider the set $Q$ of light colors of $\pal_v$ that appear infrequently in $C'_v$, or in at most half of the palettes of $C'_v$: $Q = \{c \in \pal_v\setminus \hcol_v: |\Psi^{-1}(c) \cap C'_v| \le |C'_v|/2 \}$, where $\pal^{-1}(c)=\{u : c\in \pal_u\}$. 
Let $M$ denote the number of times a color of $\pal_v$ is missed by a node in $C'_v$: $M = \sum_{c \in \pal_v} |C'_v \setminus \Psi^{-1}(c)|$. 
By the definition of $C'_v$, $M \le |C'_v| \cdot \ehc e_v/20$, but by definition of $Q$, $M \ge |Q| \cdot |C'_v|/2$.
Thus, $|Q| \le \ehc e_v/10$. 
Since the colors in $Q$ are light, \[
\sum_{w\in E_v}\frac{|\pal_w\cap Q|}{|\pal_w|}=\sum_{c\in Q}\sum_{w\in E_v: c\in \pal_w}\frac{1}{|\pal_w|}= \sum_{c \in Q} H(c) \le |Q|/\ehc \le e_v/10\ .
\]
Hence, at most half of the nodes of $E_v$ have more than one fifth of their colors from $Q$. Removing this set from $E'_v$ results in $E''_v$ of size at least $|E''_v| \ge |E'_v| - |E_v|/2 \ge (1-4\ehat - 1/2)e_v$.

We claim that the nodes in $E''_v$ are all good. Let $w \in E''_v$. 
Since at most one-fourth of $w'$ palette is in $\hcol_v$ (since $w\notin E'_v$) and at most one-fifth is in $Q$, at least a half is light and outside $Q$: 
$|\pal_w\setminus (\hcol_v \cup Q)| \ge |\pal_w|/2$.
By the definition of $Q$, these colors are contained in at least half of the palettes of $C'_v$, hence:
\[ \sum_{u \in C'_v} |\pal_w \cap \pal_u \setminus \hcol_v| \ge \frac{|\pal_w|}{2} \cdot \frac{|C'_v|}{2}\ . \]
The node $w$ has at most $\eacd (1+\eacd) |C_v|$ neighbors in $C'_v$, since it is of lower degree than $v$ and belongs to a different almost-clique.
Subtracting these neighbors we get that
\begin{align*}
\sum_{u \in C'_v \setminus N(w)} |\pal_u \cap \pal_w \setminus \hcol_v| 
& \ge \sum_{u \in C'_v} |\pal_u \cap \pal_w \setminus \hcol_v| - |C'_v \setminus N(w)| \cdot |\pal_w| \\
& \ge |\pal_w| (|C'_v|/4 - \eacd (1+\eacd) |C_v|) \\
& \ge |\pal_w| |C_v|/120\ ,
\end{align*}
using that $|C'_v| \ge |C_v|/15$ and $\eacd \le 1/125$.
\end{proof}

Recall that an unconnected node pair $u,w \in N(v)$ forms a \emph{same-colored light pair} if they are assigned the same light color and no other node in $N(v)$ is also assigned that color.

\begin{claim}\label{claim:goodnodelightcolor}
Let $w$ be a good node. The expected number of nodes in $C'_v$ that are assigned the same light color as $w$ is $\Omega(1)$.
\end{claim}
\begin{proof}
The probability that $w$ and a given node $u\in C'_v\cap \overline{N}(w)$ try the same light color is $|\pal_u \cap \pal_w \setminus \hcol_v|/(|\pal_w||\pal_u|) \ge |\pal_u \cap \pal_w \setminus \hcol_v|/(2|\pal_w||C_v|)$. 
Hence, the expected number of nodes $u$ with which $w$ tries the same light color is $\sum_{u\in C'_v\cap \overline{N}(w)}|\pal_u \cap \pal_w \setminus \hcol_v|/(2|\pal_w||C_v|) \ge 1/240$, where the last equality uses the definition of a good node. 

\end{proof}

\begin{claim}
The set $S_v$ of good neighbors generate $\Omega(e_v)$ slack for $v$ w.p.\ $1-\exp(-\Omega(e_v))$.
\label{claim:good-slack}
\end{claim}
\begin{proof}
For each light color $c$ let $A_c$ be the event that some non-edge $uw$ with $w\in S_v$ and $u\in C'_v$ tries the color $c$. Given that the endpoints of a non-edge $uw$ try some light color $c$, the probability that no other node in $N^+(w)\cup N^+(u)\cup N(v)$ tries $c$ is $\Omega(1)$, due to \cref{lem:light}. Letting $B_c$ be the event that $c$ is successfully tried, we get two families of events that satisfy \cref{lem:talagrand-difference}: $f=\sum_{\lcol} \mathbb{I}_{A_c}$ and $g=\sum_{\lcol} \mathbb{I}_{A_c\cap \overline{B}_c}$ can easily be shown to be $1$-Lipschitz and $3$-certifiable functions of the color trials of all nodes. By \cref{claim:goodneighborsdense}, \cref{claim:goodnodelightcolor}, and \cref{lem:light}, $\Exp[f]=\Omega(e_v)$ and  $\Exp[h]=\Exp[f-g]=\Omega(e_v)$, and \cref{lem:talagrand-difference} therefore implies the claimed slack.
\end{proof}
The lemma now follows from \cref{claim:good-slack}.
\end{proof}

\section{Palette Sparsification}
\label{sec:palette}

The original sparsification theorem of \cite{ACK19} was for $\Delta+1$-colorings. It was generalized to $\deg+1$-coloring by \cite{AA20} (as well as an approximate version for $\degoLC$). We extend their results to $\degoLC$, but using larger sample size ($O(\log^2 n)$, instead of $O(\log n)$).

The more formal version of \cref{thm:palettesparsify} is the following (nearly verbatim from \cite{AA20}):
\begin{theorem}
Let $G(V,E)$ be any $n$-vertex graph and assume each vertex $v\in V$ is given a list $\pal_v$ of $d_v+1$ colors. Suppose for every vertex $v \in V$, we independently sample a set $L(v)$ of colors of size $\ell=\Theta(\log^2 n)$ uniformly at random from colors in $\pal_v$, then, w.h.p., there exists a proper coloring of $G$ from the lists $L(v)$ for $v\in V$.
\label{thm:sparsified}
\end{theorem}
\begin{proof}
Essentially all the parts needed for \degoLC sparsification are already in \cite{ACK19,AA20}, except for the slack generation result for sparse nodes (\cref{L:slackgen-sparse}).
We restate our slack generation result for the case of higher-degree nodes.
The reason for the higher sample-size requirement of our result is that we need higher degree lower bounds in order to show w.h.p.\ that slack is generated for heavy nodes.

\begin{proposition}
\label{L:slackgen-sparse-highdeg}
Let $J$ be the set of nodes in $\Vun \cup \Vsp\setminus \Vst$ of degree at least $C\log^2 n$, for a sufficiently high constant $C$.
After running {\slackgeneration} in the subgraph induced by $\Vsp \cup \Vun$, each node $v$ in $J$ has slack $\Omega(d_v)$, w.h.p.
\end{proposition}

We handle the nodes of degree at most $C\log^2 n$ separately, coloring them after all the other ones are colored. Since we sample their whole palette, they can be $\deg+1$-colored. 
We need then only to identify the (high-degree) nodes of $\Vst$, and ensure that we color them first. 
With that, we have the following observation.
Recall that $\pal(v_i)$ be the (current) palette of node $v_i$. 
\begin{observation}
  Suppose after $\slackgeneration$, the uncolored nodes of $\Vsp$ are ordered arbitrarily $v_1, v_2, \ldots, v_{|\Vsp|}$ so that: a) nodes of $\Vst$ are ahead of the nodes of $\Vsp\setminus \Vst$, and b) nodes of degree $O(\log^2 n)$  are colored last (within each subset). Then, it holds for every $i=1,2,\ldots, |\Vst|$ that at most $(1-\delta)|\pal(v_i)|$ neighbors of $v_i$ are ahead of $v_i$ in the ordering.
  \label{obs:ordering-sparse}
\end{observation}
With this observation, the validity of the coloring of the sparse (and uneven) nodes (Lemma 4.15 of \cite{AA20}) holds with minimal changes.
\end{proof}

\cref{thm:sparsified} immediately gives an exponential-time algorithm to find the $\degoLC$-coloring, but the coloring can in fact be computed efficiently, even in restricted models of computation. The general schema is to proceed as follows: a) Compute the ACD, b) Compute the \emph{conflict graph} $G'=(V,E')$, where nodes $u$ and $v$ are adjacent if $L(u) \cap L(v) \ne \emptyset$, c) Generate slack for the sparse nodes, d) List color the sparse nodes, and e) List color the dense nodes.
Steps a), b), and e) were shown for \degoLC in \cite{AA20} and require no modification. The conflict graph contains $O(n \log^2 n)$ edges, w.h.p., by \cite[Lemma 5.2]{AA20}. 
Step c) is given by our \cref{L:slackgen-sparse-highdeg}.
Step d) follows the identical approach as \cite{ACK19,AA20}; the only difference is the ordering of the sparse nodes, since not all of them receive permanent slack. 
We show below how to implement this ordering in the three different models of \cite{ACK19}.

\paragraph{Implementation}
To simplify the task of ordering the nodes of $\Vst$, we modify the specification of $\Vst$.
Let $\Vue$ be the set of sparse nodes that are $2\sqrt{\ehat}$-uneven.
We redefine $\Veasy$ as
\[  \Veasy = \Vbal \cup \Vue \cup \Vun \cup 
    \{ v \in \Vsp : |N(v)\cap \Vdense| \ge \ehat d_v \}\ . \]
i.e., we replace $\Vdisc$ by $\Vue$ in the definition.
$\Vst$ is defined as before as the set of nodes $v$ with at least $\ehat d_v$ neighbors in $\Veasy$.
We observe that the proof of the tough case (\cref{lem:tough-slack}) goes through equally with this definition of $\Vst$. 
We focus then on indicating how to identify the redefined nodes of $\Vst$ in the different models. 

\paragraph{\MPC}
Each node can compute its degree and forward it to its neighbors, thus detecting if it is in $\Vbal$ or $\Vue$. Since the graph has only $O(n\log^2 n)$ edges, it can be gathered at a single node in $O(1)$ rounds using Lenzen's transform \cite{Lenzen13}.

\paragraph{Streaming}
Like \cite{ACK19}, we rely on the standard primitive of $\ell_0$-samplers for sampling elements in dynamic streams (see \cref{app:ps}). 
We assume that the palettes of the nodes are given as part of the stream so that the final palette of each node appears before the first occurrence of an incident edge. The palettes can then be sampled using the $\ell_0$-sampler.

The nodes can maintain their degrees, in $\tilde{O}(n)$ space.
We sample $\Theta(\delta^2 \log n)$ edges incident on each node $v$, u.a.r., obtaining a subset $S_v$ of neighbors of $v$. The fraction of neighbors of $v$ that are above/below a given threshold $q \cdot d_v$ is within $(1\pm \delta)$-factor of the fraction of nodes in $S_v$ with that property, by the basic Chernoff bound. Thus, we identify (within error $1\pm \delta$) if each given node is balanced or uneven.
From that, using a second sample of $\Theta(\delta \log n)$ neighbors of each node, we can identify within $1\pm 3\delta$-factor if a node should be in $\Vst$, i.e., if it has enough balanced or $\ehat$-uneven neighbors.
The resulting ordering of the nodes then consists of $\Vst$ (in any order), followed by the nodes of $\Vsp \setminus \Vst$ (in any order), and finally the nodes of degree at most $C\log^2 n$.

\paragraph{Query model}
It is important to specify the data model, which differs from the $\Delta+1$-coloring and $\degoC$ since the palettes are now part of the input.
We therefore add an operation that involves querying the palettes. There are two types of queries: a) what is the $i$-th neighbor of a given vertex $v$, and b) what is the $i$-th color in the palette of $v$. We assume that querying the $i$-th neighbor with $i>d_v$ or the $i$-th color with $i>\card{\pal_v}$ is allowed, and returns a special symbol $\bot$.
With binary search, one can also obtain the degree and palette size of each node in $O(n\log n)$ queries.

The second type of queries allow us to produce the palette sample $L(v)$ of each node $v$. From that, we can produce all the edges of the conflict graph.
The first type of queries allows us to sample $\Theta(\log n)$ edges incident on each given node, giving a subset $S_v$ of neighbors of $v$. We can proceed as in the streaming setting to determine which nodes are in $\Vst$.

The process can be made non-adaptive using the same modifications as in \cite{AA20}.

\appendix

\section{Concentration Bounds}
\label{app:concentration}

\begin{lemma}[Chernoff bounds]\label{lem:basicchernoff}
Let $\{X_i\}_{i=1}^r$ be a family of independent binary random variables with $\Pr[X_i=1]=q_i$, and let $X=\sum_{i=1}^r X_i$. For any $\delta>0$, $\Pr[|X-\Exp[X]|\ge \delta\Exp[X]]\le 2\exp(-\min(\delta,\delta^2) \Exp[X]/3)$.
\end{lemma}

\begin{lemma}[Hoeffding's inequality \cite{Hoeffding}]
\label{lem:chernoff-various-ranges}
Let $X_1 \ldots X_n$ be $n$ independent random variables distributed
in $\range{a_i,b_i}$, $X := \sum_{i=1}^n X_i$ their sum. For $t >0$:
\[
\Pr\event{\abs{X-\Exp[X]} > t} \leq 2\exp \parens*{ - \frac {2 \cdot t^2} {\sum_i (b_i - a_i)^2} }\ .\]
\end{lemma}

We use the following variants of Chernoff bounds for dependent random variables. The first one is obtained, e.g., as a corollary of Lemma 1.8.7 and Thms.\ 1.10.1 and 1.10.5 in~\cite{Doerr2020}.

\begin{lemma}[Martingales \cite{Doerr2020}]\label{lem:chernoff}
Let $\{X_i\}_{i=1}^r$ be binary random variables, and $X=\sum_i X_i$.
    If $\Pr[X_i=1\mid X_1=x_1,\dots,X_{i-1}=x_{i-1}]\le q_i\le 1$, for all $i\in [r]$ and $x_1,\dots,x_{i-1}\in \{0,1\}$ with $\Pr[X_1=x_1,\dots,X_r=x_{i-1}]>0$, then for any $\delta>0$,
    \[\Pr\event{X\ge(1+\delta)\sum_{i=1}^r q_i}\le \exp\parens*{-\frac{\min(\delta,\delta^2)}{3}\sum_{i=1}^r q_i}\ .\]
    If $\Pr[X_i=1\mid X_1=x_1,\dots,X_{i-1}=x_{i-1}]\ge q_i$, $q_i\in (0,1)$, for all $i\in [r]$ and $x_1,\dots,x_{i-1}\in \{0,1\}$ with $\Pr[X_1=x_1,\dots,X_r=x_{i-1}]>0$, then for any $\delta\in [0,1]$,
    \begin{equation}\label{eq:chernoffmore}
    \Pr[X\le(1-\delta)\sum_{i=1}^r q_i]\le \exp\left(-\frac{\delta^2}{2}\sum_{i=1}^r q_i\right)\ .
    \end{equation}
\end{lemma}

A set $\{X_i\}_{i=1}^r$ of binary random variables is \emph{read-$k$} if there is a set $\{Y_j\}_{j=1}^{m}$ of $m$ independent binary random variables and subsets $\{P_i\}_{i=1}^r$ of indices, $P_i\subseteq [m]$, such that $X_i$ is a function of only $\{Y_j\}_{j\in P_i}$, $i\in [r]$, while for each $j\in [m]$, $|\{i : j\in P_i\}|\le k$. In words, each $Y_j$ influences at most $k$ variables $X_i$. 

\begin{lemma}[read-$k$ bound \cite{kread}]\label{lem:kread}
Let $\{X_i\}_{i=1}^r$ be a family  of read-$k$ binary random variables  and let $X=\sum_{i=1}^r X_i$. For any $\delta>0$,
$\Pr[\card{X - \Exp[X]} \ge \delta r]\le 2\exp(-2\delta^2 r/k)$.
\end{lemma}

A function $f(x_1,\ldots,x_n)$ is  \emph{$c$-Lipschitz} iff changing any single $x_i$ affects the value of $f$ by at most $c$, and $f$ is  \emph{$r$-certifiable} iff whenever $f(x_1,\ldots,x_n) \geq s$ for some value $s$, there exist $r\cdot s$ inputs $x_{i_1},\ldots,x_{i_{r\cdot s}}$ such that knowing the values of these inputs certifies $f\geq s$ (i.e., $f\geq s$ whatever the values of $x_i$ for $i\not \in \{i_1,\ldots,i_{r\cdot s}\}$).
\begin{lemma}[Talagrand's inequality~\cite{DP09}]
\label{lem:talagrand}
Let $\{X_i\}_{i=1}^n$ be $n$ independent random variables and $f(X_1,\ldots,X_n)$ be a $c$-Lipschitz $r$-certifiable function; then for $t\geq 1$,
\[\Pr\event*{\abs*{f-\Exp[f]}>t+30c\sqrt{r\cdot\Exp[f]}}\leq 4 \cdot \exp\parens*{-\frac{t^2}{8c^2r\Exp[f]}}\]
\end{lemma}

The following lemma 
implies the two standard arguments used regarding slack generation: based on discrepancy and based on successful non-edges.

\begin{lemma}
\label{lem:talagrand-difference}
Let $\set*{X_i}_{i=1}^n$ be $n$ independent random variables. Let $\set*{A_j}_{j=1}^k$ and $\set*{B_j}_{j=1}^k$ be two families of
events that are functions of the $X_i$'s. Let $f=\sum_{j\in[k]} \mathbb{I}_{A_j}$, $g=\sum_{j\in[k]} \mathbb{I}_{A_j \cap \overline{B}_j}$,\footnote{$\mathbb{I}$ denotes the indicator random variable of an event.} and $h=f-g$ be such that $f$ and $g$ are $c$-Lipschitz and $r$-certifiable w.r.t.\ the $X_i$'s, and $\Exp[h] \geq \alpha \Exp[f]$ for some constant $\alpha \in (0,1)$. Let $\delta \in (0,1)$. Then for $\Exp[h]$ large enough:
\[\Pr\event*{\abs*{h - \Exp[h]} > \delta \Exp[h]} \leq \exp(-\Omega(\Exp[h]))\]
\end{lemma}
\begin{proof}
First, for any $\lambda\in (0,1)$ and $t>0$, we have 
\[\Pr\event*{\abs*{h - \Exp[h]} > 2t} \leq \Pr\event*{\abs*{f - \Exp[f]} > 2\lambda t} + \Pr\event*{\abs*{g - \Exp[g]} > 2(1-\lambda) t}.\]

Suppose for now that $\lambda t \geq 30 c\sqrt{r\Exp[f]}$ and $(1-\lambda) t \geq 30 c\sqrt{r\Exp[g]}$. Then by applying Talagrand's inequality (\cref{lem:talagrand}) we get:
\[\Pr\event*{\abs*{f - \Exp[f]} > 2\lambda t} \leq 4\exp\parens*{-\frac{\lambda^2 t^2}{8c^2 r \Exp[f]}},
\>\>\textrm{and}\>\>
\Pr\event*{\abs*{g - \Exp[g]} > 2(1-\lambda) t} \leq 4\exp\parens*{-\frac{(1-\lambda)^2 t^2}{8c^2 r \Exp[g]}}.
\]

The definitions of $f$, $g$, and $h$ together with that $\Exp[h] \geq \alpha \Exp[f]$ imply that $\Exp[f] \in \range*{\Exp[h], \frac 1 \alpha \Exp[h]}$ and $\Exp[g] \in \range*{0, \frac{1-\alpha}{\alpha}\Exp[h]}$. Setting $t=\delta\Exp[h]/2$, $\lambda= 1 / 2$, we get:
\[\Pr\event*{\abs*{h - \Exp[h]} > \delta \Exp[h]} \leq 8 \exp\parens*{-\frac{\alpha\delta^2}{128 c^2 r } \Exp[h] }.\]

This holds when $\Exp[h] \geq \frac{14000 \cdot c^2 r}{\delta^2 \alpha}$, as this implies the lower bounds on $\lambda t$ and $(1-\lambda) t$ we assumed previously.
\end{proof}

\section{Proof of \texorpdfstring{\cref{lem:slackcolor}}{Lemma~\ref{lem:slackcolor}}}
\label{app:slackcolor}

\slackcolorlemma*

A possible algorithm yielding \cref{lem:slackcolor} is described below (\cref{alg:slackcoloring}). It uses a procedure called $\multitrial(x)$ in which nodes get colored with probability $1-2^{-x}$ by (simultaneously) trying $x$ colors when their palette size to degree ratio $\card{\pal(v)}/d(v)$ is larger than $2x$ (see \cref{lem:multitrial-success}). Such a high probability of success allows us to show that by repeatedly applying $\multitrial(x)$ with rapidly increasing $x$, the uncolored degree of nodes decreases super-exponentially (see \cref{lem:slackcolor-tower}), implying the claimed runtime of the algorithm. A technical issue to solve along the way is to show (see \cref{lem:slackcolor-init}) that after a few initial random color trials, the slack of nodes increases by a constant factor so as to make them eligible for application of \cref{lem:multitrial-success}. Finally, when the degrees become too small, we can no longer increase the number of colors tried as fast as before, so we need to apply {\multitrial} with slower pace to finish coloring (see \cref{lem:slackcolor-finish}). 
In what follows, let $2\knuthupuparrow i$ be the iterated exponentiation (or tetration), defined by $2\knuthupuparrow 0 = 1$ and $2\knuthupuparrow (i+1) = 2^{2\knuthupuparrow i}$.

\begin{algorithm}[H]\caption{\slackcolor[($\smin$)], for node $v$} 
\label{alg:slackcoloring}
  \begin{algorithmic}[1]
  \STATE \algorithmicfor\ $O(1)$ rounds \algorithmicdo\  {\tryrandomcolor}($v$).\label{step:slackcolor-begin-init} 
    \STATE \algorithmicif\ $s(v) < 2d(v)$ \algorithmicthen\ terminate.\label{step:slackcolor-end-init}
    \STATE Let $\sminpow\gets \smin^{1/(1+\kappa)}$
    \FOR{$i$ from $0$ to $ \log^* \sminpow$}\label{step:slackcolor-begin-tower}
    \STATE $x_i \gets 2 \knuthupuparrow i$ 
    \STATE $\multitrial(x_i)$ 2 times.
    \STATE \algorithmicif\ $d(v) > s(v) / \min(2^{x_i},\sminpow^{\kappa})$ \algorithmicthen\ terminate.\label{step:slackcolor-termtower}
    \ENDFOR\label{step:slackcolor-end-tower}
    \FOR{$i$ from $1$ to $\ceil*{1/\kappa}$}\label{step:slackcolor-begin-finish}
    \STATE $x_i \gets \sminpow^{i \cdot \kappa}$ 
    \STATE $\multitrial(x_i)$ 3 times.
    \STATE \algorithmicif\ $d(v) > s(v) / \min(\sminpow^{(i+1)\cdot\kappa},\sminpow)$ \algorithmicthen\ terminate.\label{step:slackcolor-termfinishloop}
    \ENDFOR
    \STATE $\multitrial(\sminpow)$.\label{step:slackcolor-end-finish}
\end{algorithmic}
\end{algorithm}

\begin{algorithm}[H]\caption{{\multitrial}($x$), for node $v$}
\label{alg:multitrial}
  \begin{algorithmic}[1]
    \STATE $v$ picks a set $X_v$ of $x$ random colors in its palette $\pal_v$, sends them to its neighbors.
    \IF {$\exists \col \in X_v$ s.t. $\forall u\in N(v)$, $\col \not \in X_v$} 
    \STATE Adopt some such $\col$ as permanent color and broadcast to $N(v)$.
    \ENDIF
\end{algorithmic}
\end{algorithm}

\begin{lemma}
\label{lem:multitrial-success}
For every node $v$, if $x \leq \card{\pal(v)}/2\card{N(v)}$, then an execution of {\multitrial}$(x)$ colors $v$ with probability $1-2^{-x}$, even when conditioned on any particular combination of random choices of the other nodes.
\end{lemma}
\begin{proof}
    Consider $Y_v=\bigcup_{u\in N(v)} X_u$, the set of colors tried by neighbors of $v$.
    Note that $|Y_v|\le x|N(v)|\le |\pal(v)|/2$, and its composition is independent from $v$'s choice of random colors. Hence each node tried by $v$ has a probability at least $1/2$ of not being in $Y_v$, so $v$ gets colored w.p.\ at least $1-2^{-x}$.
\end{proof}

In the proofs of the coming statements, when analyzing the effect of running a procedure on the degree of a node $v$, let $d(v)$ be the degree of $v$ before applying the procedure, while $d'(v)$ is the its degree afterwards.

\begin{lemma}
    Let $\iratio > 1$. Suppose all nodes satisfy $s(v) \geq  d(v)/\iratio$. Then after $t=O(\iratio \log \iratio)$ iterations of all nodes running \tryrandomcolor, a node $v$ satisfies $s(v) \geq 2d(v)$ w.p.\ $1-\exp(-\Omega(s(v)))$. This holds conditioned on arbitrary random choices of nodes at distance $\geq 2$ from $v$.
\label{lem:slackcolor-init}
\end{lemma}
\begin{proof}
    Due to slack, each color try succeeds w.p.\ at least $p_\iratio=(1/\iratio)/(1+1/\iratio)=1/(1+\iratio)$ regardless of the random choices of other nodes. Notably, each color try in $v$'s neighborhood succeeds with at least this probability, regardless of the random choices at distance $\geq 2$ from $v$. In $t$ iterations of {\tryrandomcolor}, each node stays uncolored w.p.\ at most $(1-p_\iratio)^{t}$, hence in expectation, $(1-p_\iratio)^t d(v)$ neighbors of $v$ stay uncolored. Setting $t=\iratio\ln(4\iratio)$
    implies $(1-p_\iratio)^{-t}= (1+1/\iratio)^{t}\ge  4\iratio$, and with $\delta=(1-p_\iratio)^{-t}s(v)/(2d(v))-1$, we have $\delta\geq 4\iratio \cdot 1/(2\iratio)-1= 1$. The lemma then follows by \Cref{lem:chernoff}: 
    \[\Pr\event*{d'(v) \geq \frac {s(v)} 2} = \Pr\event*{d'(v) \geq (1+\delta)(1-p_\iratio)^{t} \cdot d(v)} \leq \exp\parens*{-\frac \delta 3 \cdot (1-p_\iratio)^{t} \cdot d(v)} \leq e^{-s(v)/12} \ .\qedhere\]
\end{proof}

\begin{lemma}
    Let $v$ be a node and $x\geq 1$ be an integer. Suppose $d(u) \leq s(u) / x$ and $s(u) \geq \smin$ for all $u\in N(v)\cup\set{v}$. Let $y\geq s(v) \cdot 2^{-x}$. Then after $t=12$ iterations of \multitrial[$(x)$], $v$ satisfies $d'(v) \leq y$ w.p.\ $1-\exp(-\Omega(y))-O(\nu\cdot \Delta)$, where $\nu \leq e^{-\Omega(\smin)}+n^{-\Theta(1)}$. This holds conditioned on arbitrary random choices of nodes at distance $\geq 2$ from $v$.
    \label{lem:slackcolor-tower}
\end{lemma}
\begin{proof}
    First, running \multitrial[$(x)$] $2$ times makes a node get colored w.p.\ at least $1-2^{-2x}$ by \cref{lem:multitrial-success}. This implies $\Exp\event{d'(v)}\le 2^{-2x}d(v)\le (2^{-2x}/x)s(v)\le y/2$.
    Applying \Cref{lem:chernoff} with $\delta=(y/\Exp\event{d'(v)}) - 1 \geq 1$, we get:
    \[\Pr\event*{ d'(v) > y } = \Pr\event*{ d'(v) > (1+\delta)\Exp\event{d'(v)} }\leq \exp\parens*{-(\delta/3)\Exp\event{d'(v)} } = \exp\parens*{-\Omega(y)} \]
    
    Therefore, a node that -- together with its neighborhood -- satisfies $d(v) \leq s(v) / x$, satisfies $d'(v) \leq y$ w.p.\ at least $1-\exp(-\Omega(y))$ after $2$ iterations of \multitrial[$(x)$]. Since $\multitrial$ succeeds with the claimed probability regardless of the random choices of a node's neighbors, the lemma holds for arbitrary random choices at distance $\geq 2$ from $v$. 
\end{proof}

\begin{lemma}
    Consider a node $v$ and integers $\smin$ and $x\geq \ln(d(v))$ such that each of $v$'s neighbors $u$ satisfies $s(u) \geq x \cdot d(u)$ and $s(u) \geq \smin$.
    Then for every $y\ge 1$, after $3$ iterations of \multitrial[$(x)$], $d'(v) \leq y/x$ w.p.\ $1-e^{-y}$. This holds conditioned on arbitrary random choices of nodes at distance $\geq 2$ from $v$.
\label{lem:slackcolor-finish}
\end{lemma}
\begin{proof}
     By \cref{lem:multitrial-success}, after $3$ iterations of \multitrial[$(x)$], each neighbor $u$ of $v$ stays uncolored w.p.\ at most $e^{-2x}$. This holds even conditioned on arbitrary random choices from $u$'s neighbors (and so of nodes at distance at least $2$ from $v$). Thus, for a specific set of $k\le d(v)$ neighbors of $v$, with the same conditioning, the probability that they all stay uncolored is bounded by $e^{-2k\cdot x}$ (using the chain rule).
     The probability that $k$ or more neighbors of $v$ stay uncolored is bounded by $\binom{d(v)}{k} \cdot e^{-2k\cdot x} \leq \exp(k\cdot (\ln d(v)-2x)) \leq e^{-k\cdot x}$. So, $d'(v) \leq y/x$ holds w.p.\ at least $1-e^{-y}$.
\end{proof}

\begin{proof}[Proof of \cref{lem:slackcolor}]
    After the first loop of \cref{alg:slackcoloring}, 
    by \cref{lem:slackcolor-init}, each node satisfies $s(v)\geq 2d(v)$ w.p.\ $1-\exp(-\Omega(s(v))) \geq 1-\exp(-\Omega(\smin))$. After step~\ref{step:slackcolor-end-init}, all non-terminated nodes $v$
    satisfy $s(v)\geq 2d(v)$. Let $\sminpow=\smin^{1/(1+\kappa)}$, as in the algorithm. Note that for every $v$, $s(v)\ge \sminpow^{1+\kappa}$.

    Let us consider steps~\ref{step:slackcolor-begin-tower} to~\ref{step:slackcolor-end-tower}. Let $x_i=2\knuthupuparrow i$ and $y_i=s(v)/\min(2^{x_i},\sminpow^{\kappa})$ (note that $2^{x_i}=x_{i+1}$). At the beginning of the $i$th execution of the loop (starting with $i=0$), all nodes satisfy $d(v) \leq s(v) / x_i$, and by definition $y_i \geq s(v)\cdot 2^{-x_i}$. By \cref{lem:slackcolor-tower}, the following execution of \multitrial[$(x_i)$] ensures that a node $v$ passes the test at the end of the $i$th loop w.p.\ $1-\exp(-\Omega(y_i))$. A node $v$ passes all the end-loop tests w.p.\ $1-\sum_i \exp(-\Omega(y_i)) = 1-\exp(-\Omega(s(v)/\sminpow^{\kappa})) \geq 1-\exp(-\Omega(\sminpow))$. At the end of this loop, each non-terminated node $v$ satisfies  $d(v)\le s(v)/\sminpow^\kappa$.
    
    Finally, in steps~\ref{step:slackcolor-begin-finish} to~\ref{step:slackcolor-end-finish}, each loop execution decreases the degree by a multiplicative factor of $\sminpow^{-\kappa}$. More precisely, let $y_i=s(v) \cdot \sminpow^{-i\cdot \kappa}$. By \cref{lem:slackcolor-finish}, the $i$th execution (starting from $1$) starts with nodes $v$ all satisfying $d(v) \le y_i$, and ends each of them satisfying $d(v) \le y_{i+1}$ (i.e., passing the test at line \ref{step:slackcolor-termfinishloop}) w.p.\ $1-e^{-\sminpow}$. Nodes that pass all the tests (w.p.\ $\ge 1-(1/\kappa)e^{-\sminpow}=1-e^{-\Omega(\sminpow)}$, since $\kappa>1/\smin$) end up with $s(v)/d(v) \geq \sminpow$. Running \multitrial[$(\sminpow)$] at this point, each remaining node gets colored w.p.\  $1-e^{-\Omega(\sminpow)}$. In total, the probability of not getting colored (in this last step or due to an early termination) is  $e^{-\Omega(\sminpow)}$. This holds even conditioned on arbitrary random choices at distance $\geq 2$ from $v$, as all the lemmas we invoked do.
\end{proof}

\section{Improvements and Limits for High-Degree Coloring}
\label{sec:betterhideg}

We prove in this section the following improvement for the \DeloC problem.

\logstarcorollarydelta*

Suppose $\ell$ is the threshold for high-slack almost-cliques.
The lower bound requirement $d$ on the degrees of nodes in our algorithm occur in three places:  a) {\slackcolor} requires that $d \ge \ell \ge \log^{1+\delta} n$, b) the heavy color argument requires that $d \ge \ell \ge \log^2 \Delta$, and c) {\putaside} works when $d = \Omega(\log^3 \ell)$. 
In the \DeloC setting, all colors are light, $d = \Delta$, and all the slack comes from balanced sources. Therefore, we may set $\ell = \log^{1+\delta} n$. The only bottleneck is then the put-aside construction. We show below how to improve the construction to work as long as $\Delta = d =  \Omega(\ell^2)$.

This is obtained by performing the sample-and-delete task of \disjointsample more gradually,
thereby maintaining better tradeoffs between the sample size and the dependency degree needed to apply the read-$k$ concentration bound (\cref{lem:kread}).

We observe that in essence, our task  is finding an \emph{independent transversal} in a graph derived from $G$. We start by stating our result in terms of transversals, since this may be of independent interest, then explain how it applies to our coloring algorithm.
The $k$-independent transversal problem takes as input a graph $H=(V_H,E_H)$, partitioned into independent sets $I_1, I_2, \ldots, I_t$, and the objective is to find an independent set $P \subset V_H$ such that $\card{I_i \cap P} \ge k$, for all $i$.
The primary parameters besides $k$ are the maximum degree $\Delta=\Delta_H$ and the size of the smallest set $I_i$. A celebrated result of Haxell \cite{Haxell01} shows that every graph has a 1-independent transversal when $|I_i| \ge 2 \Delta$, for all $i$, and this is best possible.
We show below how to find a $k$-independent transversal in $O(1)$ rounds of \CONGEST under the assumption that $|I_i|\ge ck\Delta$, for a large enough constant $c$, and $k=\Delta^{\delta/(1+\delta)}\log n$,  for any constant $\delta\in (0,1)$, where $n=|V_H|$. We assume, for simplicity, that $\delta=1/m$ is the inverse of an integer, although the proof is easy to adapt to any rational value.

Observe the main difference of this algorithm from \disjointsample: rather than keeping only nodes with no sampled neighbors, we keep the ones with few sampled neighbors, and refine them further.

\begin{algorithm}[H]\caption{{\lowdegreesample}($P$, $q$, $B$)}
\label{alg:lowtrans}
  \begin{algorithmic}[1]
  \STATE $S\gets$ each node $v\in P$ is sampled independently w.p.\ $\pdisj = 1/(2q)$
  \RETURN $P' \gets \{v\in S : |N(v)\cap S| < B/q\}$
 \end{algorithmic}
\end{algorithm}

\begin{algorithm}[H]\caption{{\transversal}($\delta$)} 
\label{alg:transversal}
  \begin{algorithmic}[1]
    \STATE  
    $m \gets 1/\delta$, $q\gets\Delta^{\delta/(1+\delta)}$, $B_0 \gets \Delta$, $P_0\gets V_H$
    \FOR{$j=1$ to $m+1$} 
      \STATE $P_j \gets$ {\lowdegreesample}($P_{j-1}$, $q$, $B_{j-1}$)
      \STATE $B_j \gets B_{j-1}/q = \Delta^{1-j\delta/(1+\delta)}$
    \ENDFOR
    \RETURN $P_{m+1}$
 \end{algorithmic}
\end{algorithm}

\begin{lemma}\label{lem:lowdeganalysis} 
Let numbers $B,q>0$ and set $P$  of vertices be such that for every $i\in [t]$, $|P\cap I_i| \ge c qB \log n$,  for a sufficiently large constant $c$, and $|N(v)\cap P|\le B$, for every $v\in I_i$.
Let $P' =$ \lowdegreesample[$(P,q,B)$].
Then, $|P'\cap I_i| \ge |P\cap I_i|/(8q)$, w.h.p.\ for all $i\in [t]$.
\end{lemma}
\begin{proof}
    Let $S$ be the sampled set  in \lowdegreesample[($P,q,B$)], and let $I=I_i$, for some $i\in [t]$.
    Observe that by Lemma~\ref{lem:basicchernoff} with  $q_i=p_{s}=1/(2q)$, we have $|S\cap I| \ge |P\cap I|/(4q)\ge (c/4)B\log n$, w.h.p.  The remainder of the proof is conditioned on this event.
    
    Let $Y_w$, $w\in P$, be the independent indicator random variable of the event that $w \in S$, and let $X_v$, $v\in I$, be the indicator random variable of the event that  $|N(v)\cap S| \ge B/q$.
    Since $|N(v)\cap P|\le B$,
    $\Exp[|N(v)\cap S|]\le B \cdot p_{s} = B/(2q)$.
By Markov, $\Pr[X_v=1]\le \Exp[|N(v)\cap S|]/(B/q) \le 1/2$. 
Note that each variable $X_v$ is a function of independent variables $Y_w$, for $w\in N(v)$, and each $Y_w$ influences at most $|N(w)\cap P|\le B$ of the variables $X_{v}$; 
thus, for a given sample $S$, $\{X_v\}_{v\in S\cap I}$ is a read-$B$ family of random variables, and by Lemma~\ref{lem:kread}, $X_I = \sum_{v\in S\cap I}X_v \le |S\cap I|/2$ holds w.p.\  $1-\exp(-\Omega(|S\cap I|/B)) = 1 - \exp(-\Omega(c\log n))$, recalling $|S\cap I| \ge   (c/4)B\log n$.  We choose the constant $c$ large enough, so that the bound holds w.h.p.; then, at least $|S\cap I| - X_I \ge |S\cap I|/2 \ge |P\cap I|/(8q)$ nodes have  degree at most $B/q$ in $S$, as claimed. 
\end{proof}

\begin{theorem}
Let $\delta\in (0,1)$. Consider an instance $H$ with a partition $\{I_i\}_{i\in [t]}$, where for all $i\in [t]$, $|I_i|\ge ck\Delta$, for a large enough constant $c$, and $k\ge \Delta^{\delta/(1+\delta)}\log n$.
Then $\transversal(\delta)$  
returns a $k$-independent transversal, w.h.p. 
\end{theorem}

\begin{proof}
Let $P = P_{m+1}$ be the set output by \transversal, and let $I=I_i$, for some $i$.
The last iteration, $m+1$, has $B_{m+1}/q = \Delta^0 = 1$. Thus, by construction, $P$ is a transversal. To prove the size bound, we apply
Lemma \ref{lem:lowdeganalysis} and the union bound to get $|P_j\cap I| \ge |P_{j-1}\cap I|/(4q)$, for each $j=1, 2, \ldots, m+1$, and thus 
\[
|P\cap I| = |P_{m+1}\cap I| \ge \frac{|P_0\cap I|}{(4q)^{m+1}} = \frac{|I|}{(4q)^{m+1}}\ge \frac{ck\Delta}{4^{1+1/\delta}\Delta}=(c4^{-1-1/\delta})k\ .
\]
To apply Lemma~\ref{lem:lowdeganalysis}, we need $|P_m\cap I|\ge c' qB_m\log n$, for a large enough $c'$. Note that $B_m=q$, while the calculation above shows that $|P_m\cap I|\ge (c4^{-1-1/\delta})kq=(c4^{-1-1/\delta})qB_m\log n$.
\end{proof}

To apply this to our coloring setting, we let $H$ be the subgraph of $G$ induced by $\bigcup_{C: \zeta_C\le B_0} \core_C$, where the union is over all almost-cliques with sparsity $\zeta_C\le B_0=O(\log^{1+\delta} n)$, and we remove all edges within each $C$. Thus, we have the correspondence $I_i\gets \core_{C_i}$, where $C_i$ is the $i$th such almost-clique, 
and the degree of a node in $H$ is (at most) its external degree in $G$. Since we apply the procedure to almost-cliques $C$ with $\zeta_C=O(\log^{1+\delta} n)$, the latter also bounds the external degree of nodes, that is, the degree in $H$. We let $k=\Theta(\log^{1+\delta} n)$, and so we only need $|I_i|=\Omega(k\cdot \log^{1+\delta} n)=\Omega(\log^{2+2\delta} n)$. Thus \transversal allows us to sample put-aside sets $P_C$ of size $\Omega(\log^{1+\delta} n)$ in cliques of sparsity $O(\log^{1+\delta} n)$ when the maximum degree $\Delta$ of $G$ is $\Omega(\log^{2+2\delta} n)$. Replacing {\putaside} by this alternative procedure in Alg.~\ref{alg:logstar-dense} is the only modification to the algorithm.

We state as conclusion the following improvement of \cref{thm:main}.

\begin{corollary}
\label{cor:log-star}
There is a randomized \CONGEST $\Delta+1$-coloring algorithm with runtime $O(\log^* n)$, for graphs with $\Delta =\Omega(\log^{2+\delta} n)$, for any constant $\delta>0$.
\end{corollary}

\paragraph{Limitation result} 
The question if the degree lower bound of \cref{thm:main} can be further decreased is open. We note here that our transversal construction is nearly tight. 
We show via the probabilistic method that there is a graph $H$ with a partition $V_H=I_1\sqcup\dots\sqcup I_t$ such that $|I_i|=\Theta(k\Delta/\log (k\Delta))$, and $H$ has no $k$-independent transversal. 

Let $k\ge 2,t\ge 2,\Delta\ge 64+12\ln(kt)$ be integers. To construct $H$, let $|I_i|=D$, where $D$ is the largest integer such that $D< k\Delta/(16\ln D)$; note that $D=\Theta(k\Delta/\ln(k\Delta))$ and $D\ge 4$. Each edge between different parts $I_i,I_j$ is drawn independently, w.p.\ $p=\Delta/(2n)$, where $n=|V_H|=Dt$. By Chernoff bound and union bound, w.p.\ $1-n e^{-\Delta/6}>1/2$,
 the degree of each node is at most $\Delta$, where we used $\ln n-\Delta/6\le \ln (Dt)-\Delta/6\le \ln(k\Delta t)-\Delta/6\le -1$.
 For every subset $S\subseteq V_H$ such that $|S\cap I_i|=k$, $1\le i\le t$, the probability that $H[S]$ contains no edges is $(1-p)^{\binom{t}{2}k^2}\le e^{-pk^2t(t-1)/2}=e^{-k^2(t-1)\Delta/(4D)}<e^{-4k(t-1)\ln D}\leq D^{-2kt}$. 
 The number of such subsets $S$ is $\binom{D}{k}^t<D^{kt}$, so by the union bound, the probability that there exists a subset $S$ with the desired property is at most $D^{-kt}\le 1/16$. Thus, w.p.\ $1-ne^{-\Delta/6}-D^{-kt}>0$,  
 $H$ has maximum degree at most $\Delta$, and contains no $k$-independent transversal; in particular, such $H$ exists.

In the setting of our coloring algorithm, this means that in order to create slack $k=\log^{1+\delta} n$ corresponding to the external degree $\Delta_H=\log^{1+\delta} n$, we must have $\Delta_G\approx D =\Omega(k\Delta_H/\log(k\Delta_H))=\Omega(\log^{2+\delta'} n)$, for some $\delta'\in (0,\delta)$.

\section{Additional Material Related to Palette Sparsification}
\label{app:ps}

For sampling in dynamic stream, we use the following standard result based on $\ell_0$-samplers.
\begin{proposition}[{\cite[Thm.~2]{JST11}}]
There exists a streaming algorithm that given a subset $P \subseteq V\times V$ of pairs of vertices and an integer $k \ge 1$ at the beginning of a dynamic stream, outputs with high probability a set $S$ of $k$ edges from the edges in $P$ that appear in the final graph (it outputs all edges if their number is smaller than $k$). The set $S$ of edges can be either chosen uniformly at random with replacement or without replacement from all edges in $P$ that appear in the final graph. The space needed by the algorithm is $O(k \cdot \log^3 n)$.
\label{prop:l0-sampler}
\end{proposition}

\bibliographystyle{alpha}
\bibliography{arxiv_v1}
\end{document}